\documentclass[12pt]{article}
\pdfoutput=1
\usepackage{enumerate}
\usepackage{fullpage}
\usepackage{graphicx}
\usepackage{mathptmx}
\usepackage{mathrsfs}
\usepackage{amssymb}
\usepackage{amsmath}
\usepackage{graphicx}
\usepackage{amsmath,amsfonts}
\usepackage{listings}
\usepackage{natbib}
\newtheorem{lemma}{Lemma}[section]
\newtheorem{theorem}{Theorem}[section]

\newtheorem{proposition}{Proposition}[section]

\newenvironment{proof}{\noindent\\ \noindent\relax{\sc
     Proof}}{{\samepage\par\nopagebreak\hbox
     to\hsize{\hfill$\Box$}}}
\newcommand{\be}{\begin{equation}} \newcommand{\ee}{\end{equation}}
\newcommand{\bd}{\begin{displaymath}} \newcommand{\ed}{\end{displaymath}}
\newcommand{\ba}{\begin{align}} \newcommand{\ea}{\end{align}}
\newcommand{\baa}{\begin{align*}} \newcommand{\eaa}{\end{align*}}
\newcommand{\ben}{\begin{enumerate}} \newcommand{\een}{\end{enumerate}}
\newcommand{\bi}{\begin{itemize}} \newcommand{\ei}{\end{itemize}}

\newcommand{\E}[1]{\operatorname{E}\left[ #1 \right]}
\newcommand{\Var}[1]{\operatorname{Var}\left[ #1 \right]}
\newcommand{\var}[1]{\operatorname{Var}\left[ #1 \right]}
\newcommand{\cov}[2]{\operatorname{Cov}\left[ #1,#2 \right]}

\newcommand{\Et}[1]{\operatorname{E}\left[ #1 \vert \mathcal{Y}_{n}\right]}
\newcommand{\Vart}[1]{\operatorname{Var}\left[ #1 \vert  \mathcal{Y}_{n} \right]}
\newcommand{\covt}[2]{\operatorname{Cov}\left[ #1,#2 \vert  \mathcal{Y}_{n}\right]}

\title{A Consistent Estimator of the Evolutionary Rate}
\author{Krzysztof Bartoszek and Serik Sagitov}

\begin{document}

\maketitle

\begin{abstract}
We consider a branching particle system where particles reproduce according to the pure birth 
Yule process with the birth rate $\lambda$, conditioned on the observed number of particles to 
be equal $n$. Particles are assumed to move independently on the real line according to the 
Brownian motion with  the local variance 
$\sigma^2$. In this paper we treat $n$ particles as a sample of related species. The spatial 
Brownian motion of a particle describes the development of a trait value of interest 
(e.g. log--body--size). We propose an unbiased estimator $R_n^2$ of the evolutionary rate 
$\rho^2=\sigma^2/\lambda$. The estimator  $R_n^2$ is proportional to the sample variance   
$S_n^2$ computed from $n$ trait values. We find an approximate formula for the standard error 
of  $R_n^2$ based on a neat asymptotic relation for the variance of  $S_n^2$.
\end{abstract}

(Keywords: Branching Brownian motion, conditioned branching process, tree--free phylogenetic
comparative method, quantitative trait evolution, Yule process)

\section{Introduction}\label{intro}
Biodiversity within a group of $n$ related species could be quantified by comparing suitable 
trait values. For some key trait values like log body size, researchers apply 
the Brownian motion model proposed by \citet{JFel1985}. It is assumed that the current trait 
values $(X_1^{(n)},\ldots,X_n^{(n)})$  have evolved from the common ancestral state $X_{0}$ as a 
branching Brownian motion 
with the local variance  $\sigma^{2}$. Given a phylogenetic tree describing the ancestral history 
of the group of species the Brownian trajectories  of the trait values for sister species are 
assumed to evolve independently after the ancestor species splits in two daughter species. 
The resulting phylogenetic sample $(X_1^{(n)},\ldots,X_n^{(n)})$ consists of identically 
distributed normal random variables with a dependence structure caused by the underlying 
phylogenetic signal.

A mathematically appealing and biologically motivated version of the phylogenetic sample model 
assumes that the phylogenetic tree behind the normally distributed trait values 
$(X_1^{(n)},\ldots,X_n^{(n)})$ is unknown. As a natural first choice to model the unknown 
species tree,  we use the Yule process with birth rate $\lambda$   \citep[see][]{GYul1924}.  
Since  the phylogenetic sample size is given, $n$, the Yule process should be conditioned 
on having $n$ tips: such conditioned branching processes have received significant attention in 
recent years,
due to e.g. \citet{DAldLPop2005,TGer2008a,AMooetal,TreeSim1,TreeSim2,TStaMSte2012}. 
This "tree-free" approach for comparative phylogenetics was previously addressed by
\citet{SSagKBar2012} and \citet{FCraMSuc2013},  \citep[much earlier][used a 
related  branching Brownian process as a population genetics model]{AEdw1970}.

In our work we show that a properly scaled sample variance is an unbiased and consistent estimator 
of the compound parameter $\rho^2=\sigma^{2}/\lambda$ which we call the evolutionary rate of the 
trait value in question. Our main mathematical result, Theorem \ref{thmVar}, gives an asymptotical 
expression for the variance of the phylogenetic sample variance. This result leads to a simple 
asymptotic formula for the estimated standard error of our estimator. 
Our result is in agreement with the work of \citet{FCraMSuc2013} whose simulations indicate that 
their approximate
maximum likelihood procedure yields an unbiased consistent estimator of $\sigma^{2}$.
This is  illustrated using the example of the Carnivora order studied previously
by \citet{FCraMSuc2013}.

The phenotype modelled by a Brownian motion is usually interpreted as the case of 
neutral evolution with random oscillations
around the ancestral state. This model was later developed into an adaptive
evolutionary model based on the Ornstein--Uhlenbeck process by
\citet{JFel1988,THan1997,MButAKin2004,THanJPieSOrz2008,KBarJPiePMosSAndTHan2012}.
The tree-free setting using the Ornstein--Uhlenbeck process was addressed by \citet{KBarSSag2012} 
where for the Yule--Ornstein--Uhlenbeck model, some phylogenetic confidence intervals for the 
optimal trait value were obtained via three limit theorems for the phylogenetic sample mean.
Furthermore, it was shown that the phylogenetic sample variance is an unbiased consistent 
estimator of
the stationary variance of the process. 

At the end of their discussion \citet{FCraMSuc2013} write that as the 
the tree of life is refined interest in ``tree--free'' estimation
methods may diminish. They however indicate that ``tree--free''
estimates may be useful to calculate starting points for 
simulation analysis. We certainly agree with the second statement
but believe that development of ``tree--free'' methods
should proceed alongside that of ``tree--based'' ones.

One of the most useful features of the tree--free comparative models
is that they offer a natural method of tree growth allowing 
for study of theoretical properties of phylogenetic models as demonstrated
in this work \citep[and also][]{SSagKBar2012,KBarSSag2012,KBar2014,FCraMSuc2013}.
Another alternative to studying properties of these estimators is the 
tree growth model proposed by \citet{CAne2008,LHoCAne2013,CAneLHoSRoc2014}. In this setup the total height 
of the tree is kept fixed and new tips are added to randomly chosen branches. These
two approaches seem to be in agreement, at least up to the second 
moments, since e.g.
they agree on the lack of consistency of estimating $X_{0}$.
In \citet{SSagKBar2012} we showed that under the Yule Brownian
motion model $\var{\overline{X}_{n}} \to 2\sigma^{2}$. 

In a practical situation ``tree--free'' methods can be used for a 
number of purposes. Firstly as pointed out by \citet{FCraMSuc2013}
they can be useful for calculating starting points for 
further numerical estimation procedures or defining prior
distributions in a Bayesian setting. Secondly they have to be used
in a situation where the tree is actually unknown
e.g. when we are studying fossil data or trying to make
predictive statements about future phenotypes, e.g.
development of viruses. Thirdly they can be used for 
various sanity checks. If they contradict ``tree--based'' results
this could indicate that the numerical method fell into a local
maximum. 

The paper has the following structure. Section \ref{main} presents the model, 
the main results and an application. Section \ref{outline} states two lemmata and a 
proposition directly yielding the assertion of   Theorem \ref{thmVar}. 
Proposition  \ref{cypr} deals with the covariances between coalescent times for 
randomly chosen pairs of tips from a random Yule $n$-tree. 
The properties of the coalescent of a single random 
pair were studied previously by e.g. \citet{MSteAMcK2001} and \citet{SSagKBar2012}. 
In Section \ref{col} we state two lemmata needed for the proof of Proposition  
\ref{cypr}. Section \ref{prpr} contains two further lemmata and the proof of 
Proposition  \ref{cypr}. In Section \ref{pL4}, \ref{secl}, and \ref{h12} we prove the 
lemmata from Sections \ref{outline},  \ref{col}, and \ref{prpr}. 
Appendix  \ref{App} contains some useful results concerning harmonic numbers of the 
first and second order.

\section{The main results}\label{main}

The basic evolutionary model considered in this paper is characterized by four parameters 
$(\lambda, n,X_{0},\sigma^{2})$ and consists of two stochastic components: a random phylogenetic 
tree defined by parameters $(\lambda, n)$ and a trait evolution process along a lineage  defined by
parameters $(X_{0},\sigma^{2})$.
The first component, species tree connecting $n$ extant species, is modelled by the pure birth 
Yule process \citep{GYul1924} with 
the birth (speciation) rate $\lambda$ and conditioned on having $n$ tips  \citep{TGer2008a}. For the second component  we adapt the approach by assuming that for a given $i=1,\ldots, n$, the current trait value $X_i^{(n)}$ has evolved from the ancestral state $X_{0}$ according to the Brownian motion 
with the local variance  $\sigma^{2}$.

Treating the collection of the current trait values $(X_1^{(n)},\ldots,X_n^{(n)})$ 
generated by such a process as a sample of identically distributed, but dependent, 
observations, we are interested in the properties of the basic summary statistics
\begin{displaymath}
\overline{X}_{n}={X_1^{(n)}+\ldots+X_n^{(n)}\over n},\quad S^{2}_{n}
=\frac{1}{n-1}\sum_{i=1}^{n}(X_{i}^{(n)}-\overline{X}_{n})^{2},
\end{displaymath}
the sample mean and sample variance.

\begin{figure}
\begin{center}
\includegraphics[width=0.4\textwidth]{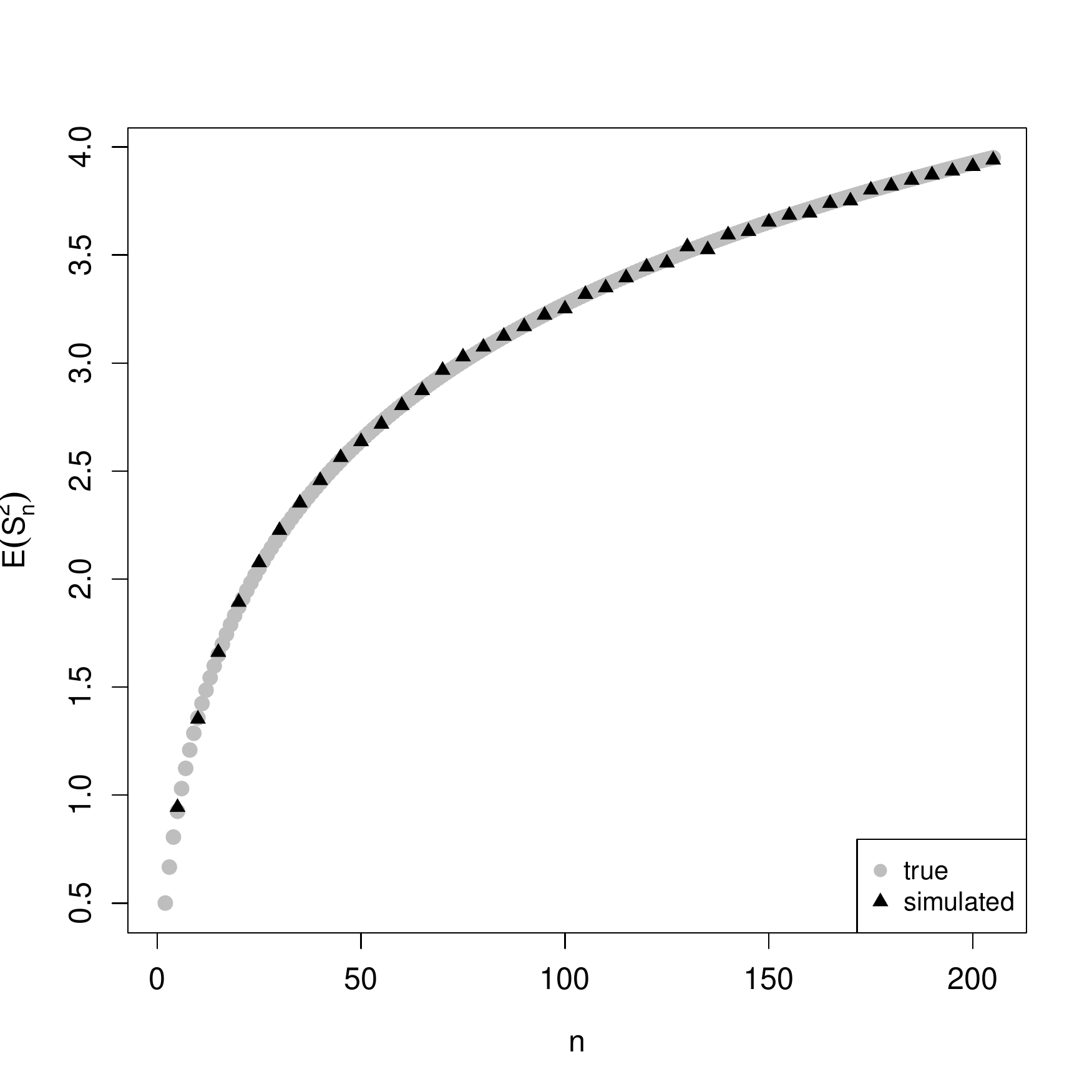}
\includegraphics[width=0.4\textwidth]{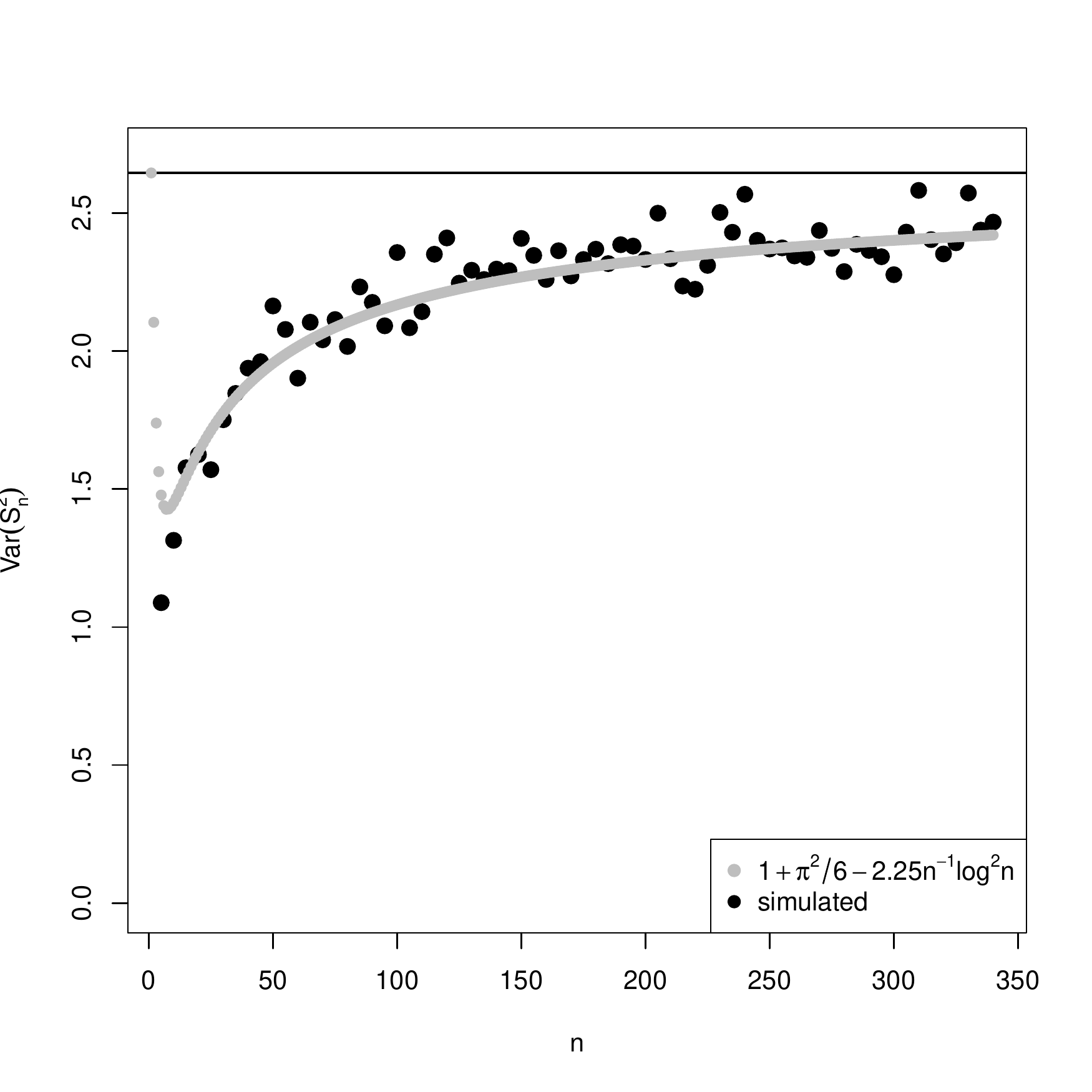}
\caption{Left: True and simulated values of 
$\E{S^{2}_{n}}$, right: simulated values of $\var{S^{2}_{n}}$ with limit
equalling $\pi^{2}/6+1$. Each point comes from $10000$ simulated Yule trees and
Brownian motions on top of them. 
Parameters used in simulations are $\lambda=1$, $X_{0}=0$ and $\sigma^{2}=1$.
The grey line on the right panel fits a curve based on the convergence rate $O(n^{-1}\log n^2)$.
\label{figMomSn}}
\end{center}
\end{figure}

According to \citep{SSagKBar2012} we have 
\[
\E{S_{n}^{2}} = \left( \frac{n+1}{n-1}H_{n} -2\frac{n}{n-1} \right){\sigma^{2}\over\lambda},
\]
see  Fig \ref{figMomSn}, left panel (all simulations are produced using the TreeSim \citep{TreeSim1,TreeSim2} and
mvSLOUCH \citep{KBarJPiePMosSAndTHan2012} R packages). It follows that the normalized sample variance
\be\label{est}
R^{2}_{n} = \left( \frac{n+1}{n-1}H_{n} -2\frac{n}{n-1} \right)^{-1}S_{n}^{2}
\ee
gives  an unbiased estimator of the compound parameter $\rho^2:={\sigma^{2}\over\lambda}$ for the Yule--Brownian--Motion model, see Fig \ref{figHistSn2}. In the comparative phylogenetics framework the ratio $\rho^2$ can be called the {\it evolutionary rate} as it measures the speed of change in the trait value when the time scale is such that we expect one speciation event per unit of time and per species. The next theorem is the main asymptotic result of this paper, illustrated by Fig \ref{figMomSn}, right panel. 
\begin{theorem}\label{thmVar}
Consider the sample variance $S^{2}_{n}$ for the Yule--Brownian--Motion model with parameters $(\lambda, n,X_{0},\sigma^{2})$. 
Its variance satisfies the following asymptotic relation

$$\var{S^{2}_{n}/\rho^2}=1+{\pi^2\over6}+O(n^{-1}\log^2n),\quad n\to\infty.$$
\end{theorem}
In terms of our estimator \eqref{est}, Theorem \ref{thmVar} yields
\[
\Var{R_{n}^{2}/\rho^2} = {1+{\pi^2\over6}\over(\log n+\gamma -2)^2}+O(n^{-1}),
\]
where $\gamma=0.577$ is the Euler constant, implying that $R_{n}^{2}$ is a consistent 
estimator of the evolutionary rate $\rho^2$. It follows that for large $n$, 
the standard error (estimated standard deviation) of the unbiased estimator $R_{n}^{2}$ 
can be approximated by 
\be\label{error}
{\rm SE } (R_{n}^{2})\approx \sqrt{1+{\pi^2\over6}}\cdot{R_{n}^{2}\over \log n+\gamma -2}\approx {1.626\over \log n-1.423}\cdot R_{n}^{2}.
\ee

\begin{figure}
\begin{center}
\includegraphics[height=0.2\textwidth,width=0.3\textwidth]{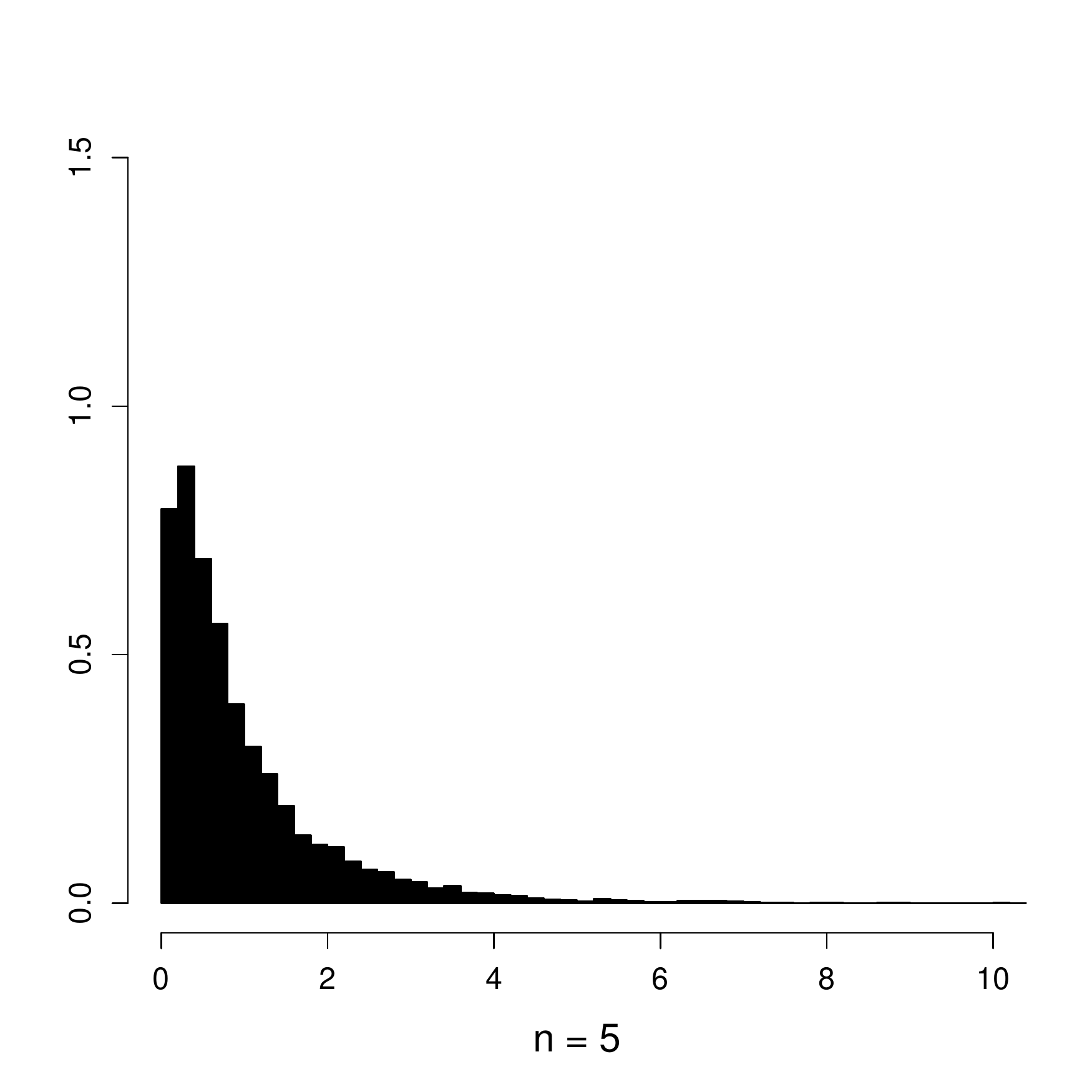}
\includegraphics[height=0.2\textwidth,width=0.3\textwidth]{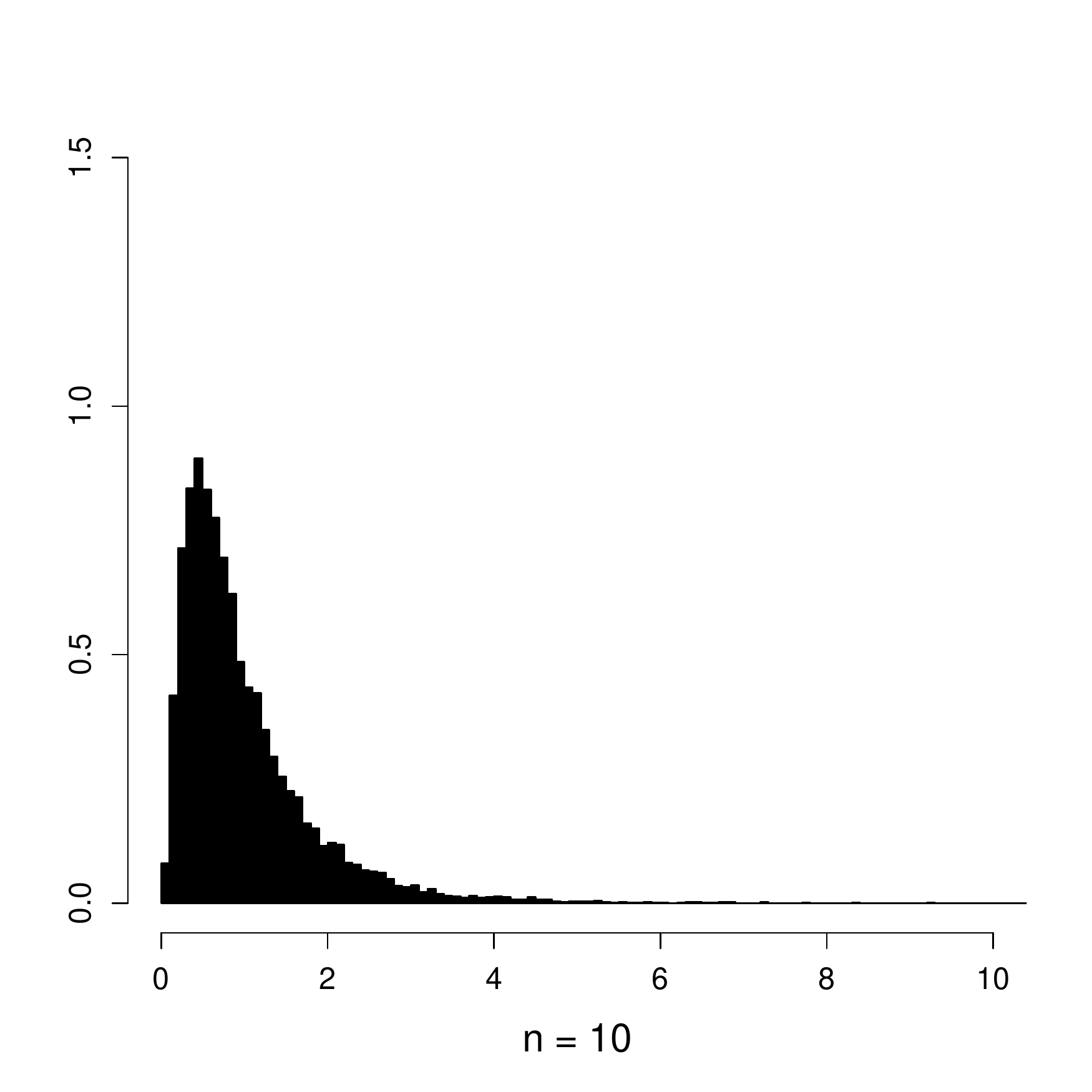}
\includegraphics[height=0.2\textwidth,width=0.3\textwidth]{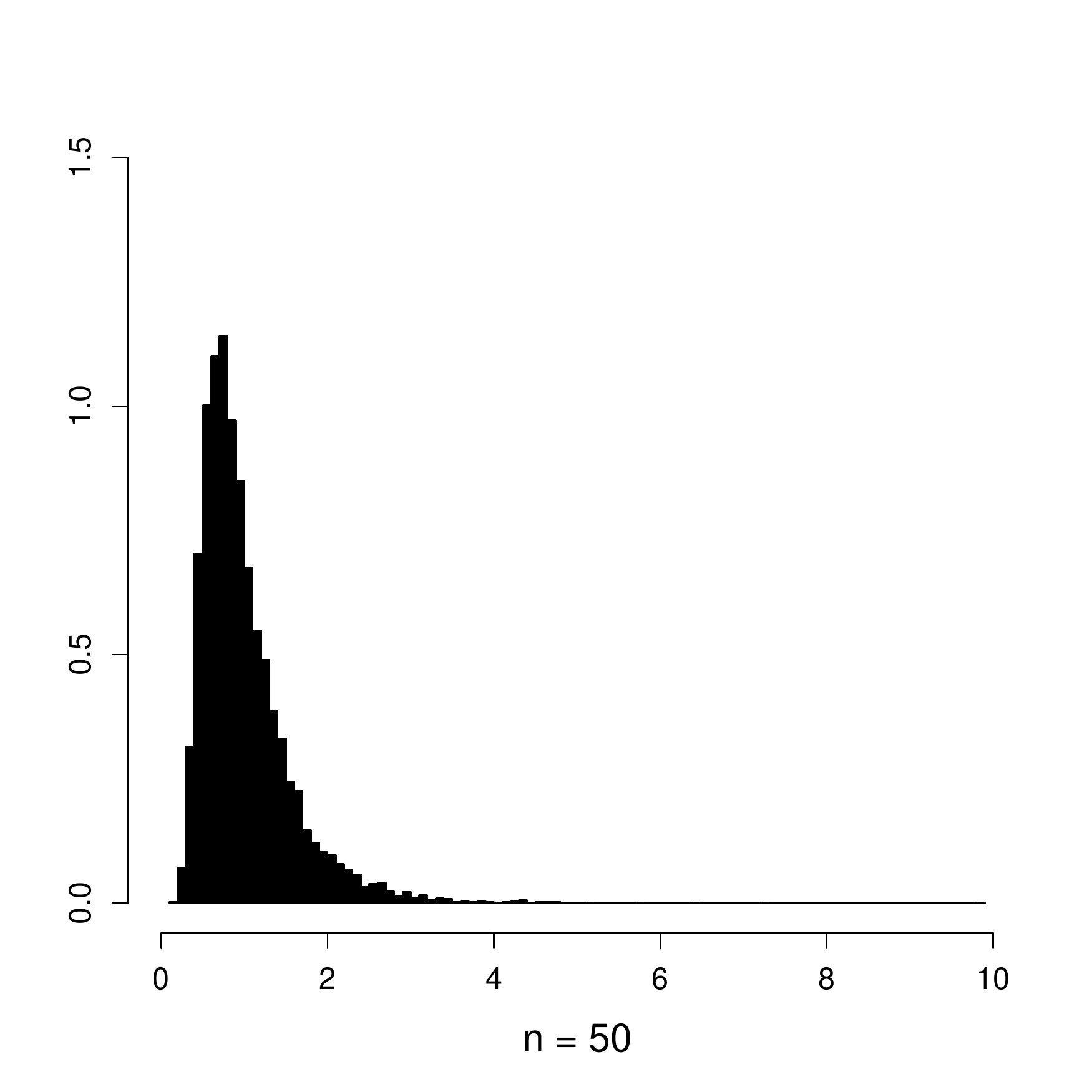} \\
\includegraphics[height=0.2\textwidth,width=0.3\textwidth]{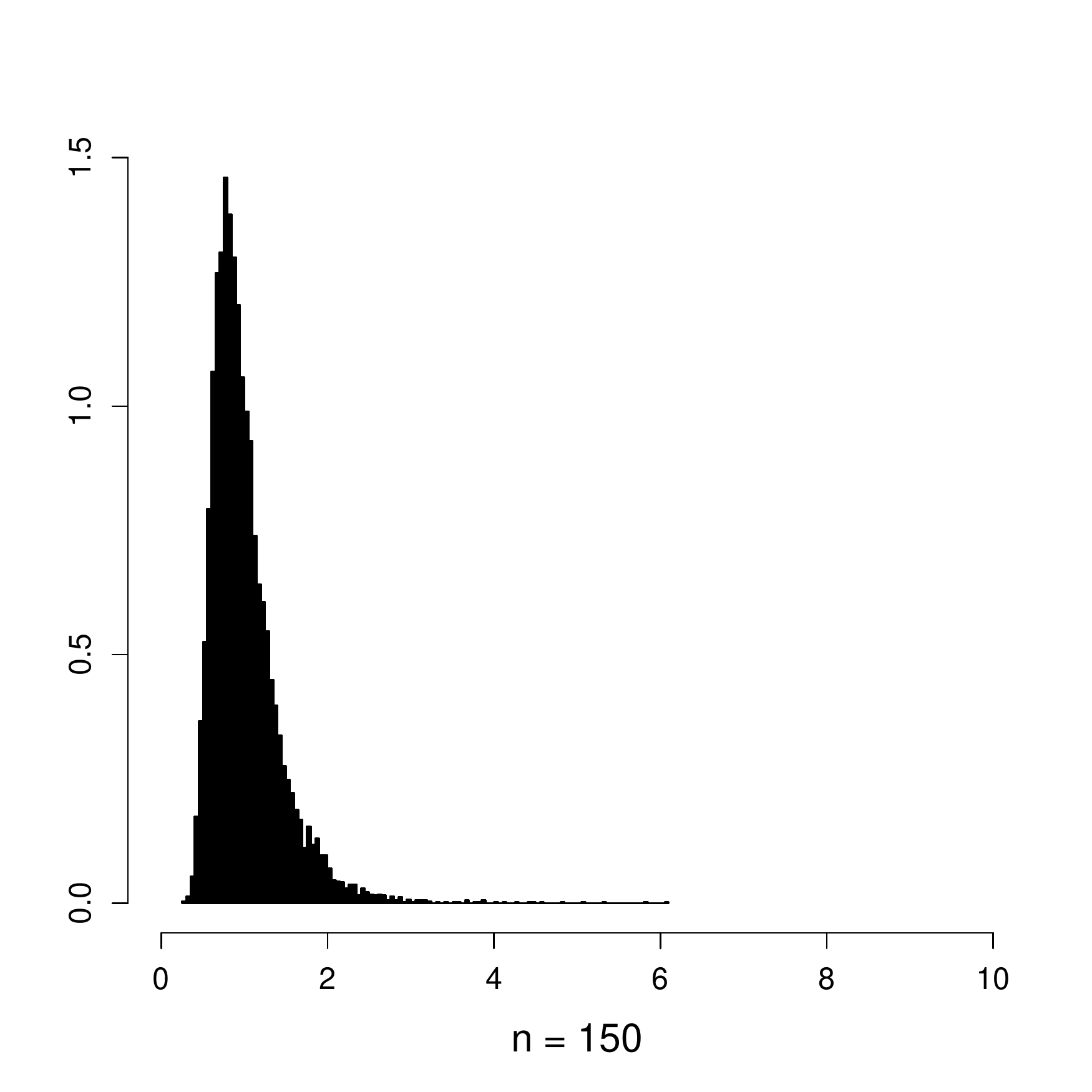}
\includegraphics[height=0.2\textwidth,width=0.3\textwidth]{HistSn2_150BM.pdf}
\includegraphics[height=0.2\textwidth,width=0.3\textwidth]{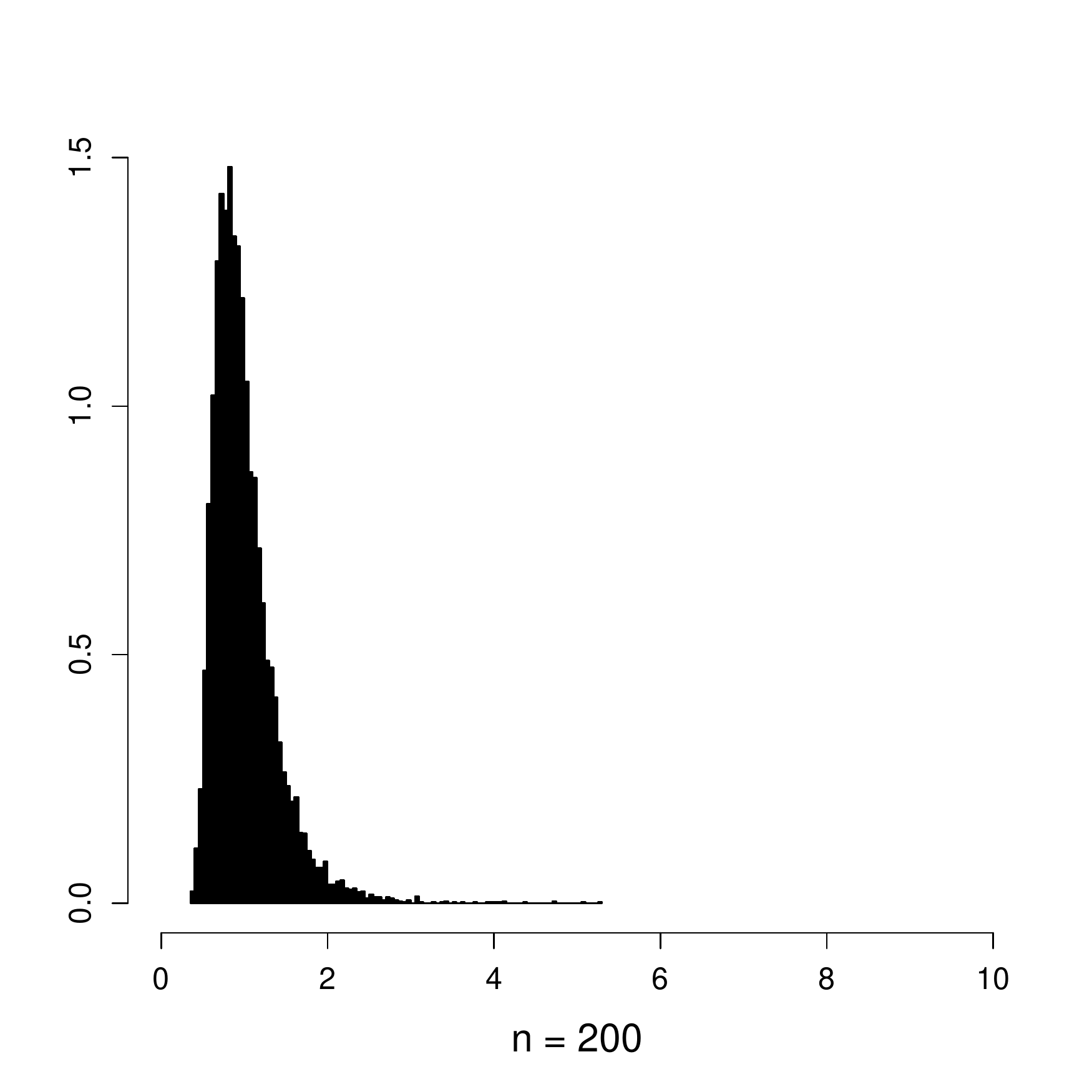} 
\caption{Histograms of 
$R^{2}_{n}$ for left to right top $n=5,10,50$ and bottom $n=100,150,200$. 
Parameters used in simulations are $\lambda=1$, $X_{0}=0$ and $\sigma^{2}=1$.
\label{figHistSn2}}
\end{center}
\end{figure}

The estimator of Eq. \eqref{est} should be compared to the  approximate maximum--likelihood  
estimator for the local variance $\sigma^{2}$ recently proposed by \citet{FCraMSuc2013} in
the same framework of the Yule--Brownian--Motion model. 
The main difference between two approaches is that in  \citet{FCraMSuc2013} it is assumed that one 
knows both the number of tips and the total height of the otherwise unknown species tree. 
The Crawford-Suchard estimator is based on a closed form of the distribution of phylogenetic 
diversity -- the sum of branch lengths
connecting the species in a clade. 

As an application of their estimator,  \citet{FCraMSuc2013}
study different families of the Carnivora
order, estimating $\sigma^{2}$ for each of the 12 clades. The data for the log-body-size disparities was taken from the PanTHERIA database \citep{PanTHERIA}. The data summary and  the Crawford-Suchard estimates are shown in the left part of 
Tab.\,\ref{tabCarnivora}. In the right part of Tab.\,\ref{tabCarnivora} we present our estimates $\hat\rho^2$ for the evolutionary rate parameter 
$\rho^2=\sigma^2/\lambda$ for each of the 12 families in the Carnivora order. The standard error is computed using \eqref{error}. We note that the data  does not take into account  the 
newly described species \textit{Bassaricyon neblina} from the Procyonidae family 
\citep{Olinguito}.

In the next-to-last column we list the ratios demonstrating a surprisingly good agreement between our and Crawford-Suchard estimates. The ratio is taken between two products: $\hat\rho^2u_n$ on one hand, and $\hat\sigma^{2}t_n$ on the other.
Here $u_n=\E{U_{n}}$ is  the expected age of the conditioned standard Yule process with $\lambda=1$, while $t_n$ is the clade age assumed to be known in the Crawford-Suchard framework. Both $\hat\rho^2u_n$ and $\hat\sigma^{2}t_n$ estimate the same quantity -- the variance in the trait values for the evolution of the corresponding clade. Therefore, one should expect these ratios to be close to one. And indeed, the 12 ratios have mean 0.97 and standard deviation 0.20. 

Our estimator and its standard error are computed by simple formulae given above. A major weakness of our estimator is relatively big standard error for realistic richness values, see the 7th column in Tab.\,\ref{tabCarnivora}. This can be explained by the fact that we do not use an additional information about the species tree, like the height of the tree used in the Crawford-Suchard estimator.

\begin{table}
\begin{center}
\begin{tabular}{ccccc|cccc}

Family & $n$ &$t_n$ &  Disparity & $\hat\sigma^{2}$ (SE)& $u_n$ & $\hat\rho^2$ (SE) &$\hat\rho^2u_n\over\hat\sigma^{2}t_n$& 
${\hat\rho^2\over\hat\sigma^{2}/\hat \lambda}$  \\
\hline
Felidae & 40 (7)& 33.3 & 1.588 &  .080  (.009)  &4.279 & .649 (.466)& 1.042            &0.560                     \\
Viverridae & 35 (6)& 37.4 &0.662 &.029  (.004) &  4.147 &  .284 (.217)&  1.086 & 0.676 \\
Herpestidae & 33 (4)& 25.5 & 0.482 &  .030   (.003)  & 4.089 & .211 (.166) &1.128  &  0.485\\
Eupleridae & 8 (0)& 25.5& 0.916 &  .079  (.010)   &2.718 & .758 (1.72)& 1.023  &0.662\\
Hyaenidae & 4 (0)& 32.2 &   0.805 &.122  (.005)  &  2.083&.999 (19.5)& 0.530 &0.565 \\
Canidae & 35 (3)& 48.9  &  0.678 & .030  (.004)  &4.147 &  .290 (.221)& 0.825      & 0.667 \\
Ursidae & 8 (0)& 42.6 & 0.303 & .024  (.002) & 2.718&  .251 (.569)&  0.667  & 0.722 \\
Otariidae & 16 (2)& 24.5 &0.386 &.028  (.003)  &  3.381 &  .227 (.274)& 1.119 & 0.559  \\
Phocidae & 19 (0)& 24.5&0.751 &  .052  (.005) & 3.548&  .410 (.438)& 1.142           &   0.544\\
Mephitidae & 12 (3)& 32.0 &0.570 &.039  (.005) & 3.103 &  .384 (.588)& 0.955 &0.679   \\
Mustelidae & 59 (10)& 27.4& 2.263 & .126  (.014)  &  4.663& .811 (.497)& 1.095  &  0.444\\
Procyonidae & 14 (1)& 27.4 &0.531 & .037  (.004)   &  3.252& .332 (.444)& 1.065 &   0.619
\end{tabular}
\caption{Data summary. 2nd column: clade richness (number of missing trait values); 3rd column: the clade age in millions of years;
4th column: ${n-1\over n}\cdot S_n^2$ trait disparity ;
6th column: the expected age $u_n=\E{U_{n}}$ of the conditioned standard Yule process with $\lambda=1$.
}\label{tabCarnivora}
\end{center}
\end{table}

This close agreement is obtained despite a number of features that complicates the comparison between two methods. Our approach in its current form does not allow to take into 
account the fact that some trait values are missing. We calculated $\hat\rho^2$ for the trait disparity 
as if it was computed using all $n$ trait values.
Moreover, it is not be clear how to take into account the
measurement variance. As shown by \citet{THanKBar2012} even with a known
tree, the measurement error  can cause very diverse effects. Therefore we would expect
the situation to be even more interesting when we integrate the phylogeny out.

In their work \citet{FCraMSuc2013} 
estimated the overall speciation rate to be $\hat\lambda=0.069$ per million years. The last column of Tab.\,\ref{tabCarnivora} demonstrates that using this common value for the speciation rate $\lambda$ produces huge discrepancy between our estimates $\hat\rho^2$ for the rates of evolution $\rho^2=\sigma^2/\lambda$ and the rates of evolution computed using the Crawford-Suchard estimates for $\sigma^2$. 
This observation points out that a fair direct comparison of  $\hat\rho^2$ and  $\hat\sigma^{2}/\hat\lambda$ would requires specific estimates of the speciation rate $\lambda$ for each of the 12 clades.


\section{Outline of the proof of Theorem \ref{thmVar}}\label{outline}
We start with a general observation, Lemma \ref{L1}, concerning the sample variance
\[D_{n}^2=\frac{1}{n-1}\sum\limits_{i=1}^{n} (Y_{i}-\overline{Y})^{2}\]
of $n$, possibly dependent and not necessarily identically distributed, 
observations  $(Y_{1},\ldots,Y_{n})$ with sample mean 
$\overline{Y}=n^{-1}\sum\limits_{i=1}^{n} Y_{i}$. 

\begin{lemma}\label{L1}
If $(W_{1},W_{2},W_{3},W_{4})$ is a random sample without replacement from random values $(Y_{1},\ldots,Y_{n})$, then
 \begin{align}\label{stist1}
\Var{D_{n}^{2}}  &=\cov{W_{1}^{2}}{W_{2}^2}-2\cov{W_{1}^{2}}{W_{2}W_3}+\cov{W_{1}W_2}{W_{3}W_4}+n^{-1} B_n,
\end{align}
where 
\begin{align*}
|B_n|  &<\E{W_1^{4}}+4\E{W_{1}^{3}W_{2}}+\E{W_{1}^{2}W_{2}^2}+ 6\E{W_{1}^{2}W_{2}W_{3}}+4\E{W_{1}W_{2}W_{3}W_{4}}.
\end{align*}
\end{lemma}
Observe that in terms of the sample variance for the scaled trait values 
\be\label{nv}
Y_i:=Y_{i}^{(n)}={X_i^{(n)}-X_0\over\sigma/\sqrt\lambda},\quad i=1,\ldots,n,
\ee
we have       $S_{n}^2={\sigma^2D_{n}^2\over\lambda}$, 
and to prove Theorem \ref{thmVar} we have to verify that 
\be\label{den}
\Var{D_{n}^{2}} =1+{\pi^2\over6}+O(n^{-1}\log^2n).
\ee
The Yule $n$-tree underlying the set of scaled values \eqref{nv} has unit speciation rate. 
We call it the {\it standard} Yule $n$-tree, and denote by $\mathcal Y_n$ be the 
$\sigma$--algebra generated by all the information describing this random tree. 
Under the Brownian motion assumption the trait values  \eqref{nv}
are conditionally normal with
\begin{align*}
\Et{Y_i}&=0,\qquad \Vart{Y_i}=U_n,
\end{align*}
where $U_n$ is the height of the standard Yule $n$-tree, see Fig. \ref{F3}. 
Moreover, see Section \ref{pL4}, we have 
\begin{align}\label{ccov}
\covt{Y_i}{Y_j}&=U_n-\tau_{ij}^{(n)},
\end{align}
where $\tau_{ij}^{(n)}$ is the backward time to the most recent common ancestor for a pair of distinct tips $(i,j)$ in the standard Yule $n$-tree, see Fig. \ref{F3}. For a quadruplet  $(i,j,k,l)$ of tips randomly sampled without replacement out of $n$ tips in the standard Yule $n$-tree, we denote

\be\label{tou}
\tau_{1}^{(n)}=\tau_{ij}^{(n)},\quad \tau_{2}^{(n)}=\tau_{ik}^{(n)},\quad \tau_{3}^{(n)}=\tau_{lk}^{(n)},\quad \tau_{4}^{(n)}=\tau_{jk}^{(n)},\quad \tau_{5}^{(n)}=\tau_{jl}^{(n)},\quad \tau_{6}^{(n)}=\tau_{kl}^{(n)}
.
\ee

\begin{lemma}\label{cyp} 
Let $(W_{1},W_{2},W_{3},W_{4})$ be a random sample without replacement of four trait values out of $n$ random values defined by \eqref{nv}.
 Then in terms of the coalescent times \eqref{tou} we have
\begin{align*}
\cov{W_1^{2}}{W_2^{2}}&-2\cov{W_1^{2}}{W_2W_3}+\cov{W_{1}W_{2}}{W_{3}W_{4}}\\
&=2\var{\tau_1^{(n)}}-4\cov{\tau_{1}^{(n)}}{\tau_{2}^{(n)}}  +3\cov{\tau_{1}^{(n)}}{\tau_{3}^{(n)}}.
\end{align*}

\end{lemma}
In view of Lemmata \ref{L1} and \ref{cyp} which are proven in Section \ref{pL4}, 
to verify \eqref{den}  it suffices to show the following asymptotic result.

 \begin{proposition}\label{cypr} 
Consider  the coalescent times \eqref{tou}. As $n\to\infty$,
 \begin{align*}
\var{\tau_1^{(n)}}&={\pi^2\over6}+O(n^{-1}\log^2 n),\\
\cov{\tau_{1}^{(n)}}{\tau_{2}^{(n)}}&=2-{\pi^2\over6}+O(n^{-1}\log^2 n),\\
\cov{\tau_{1}^{(n)}}{\tau_{3}^{(n)}}&= 3-{5\pi^2\over18}+O(n^{-1}\log^2 n).
\end{align*}
\end{proposition}
\begin{figure}
\begin{center}
\includegraphics[width=0.3\textwidth]{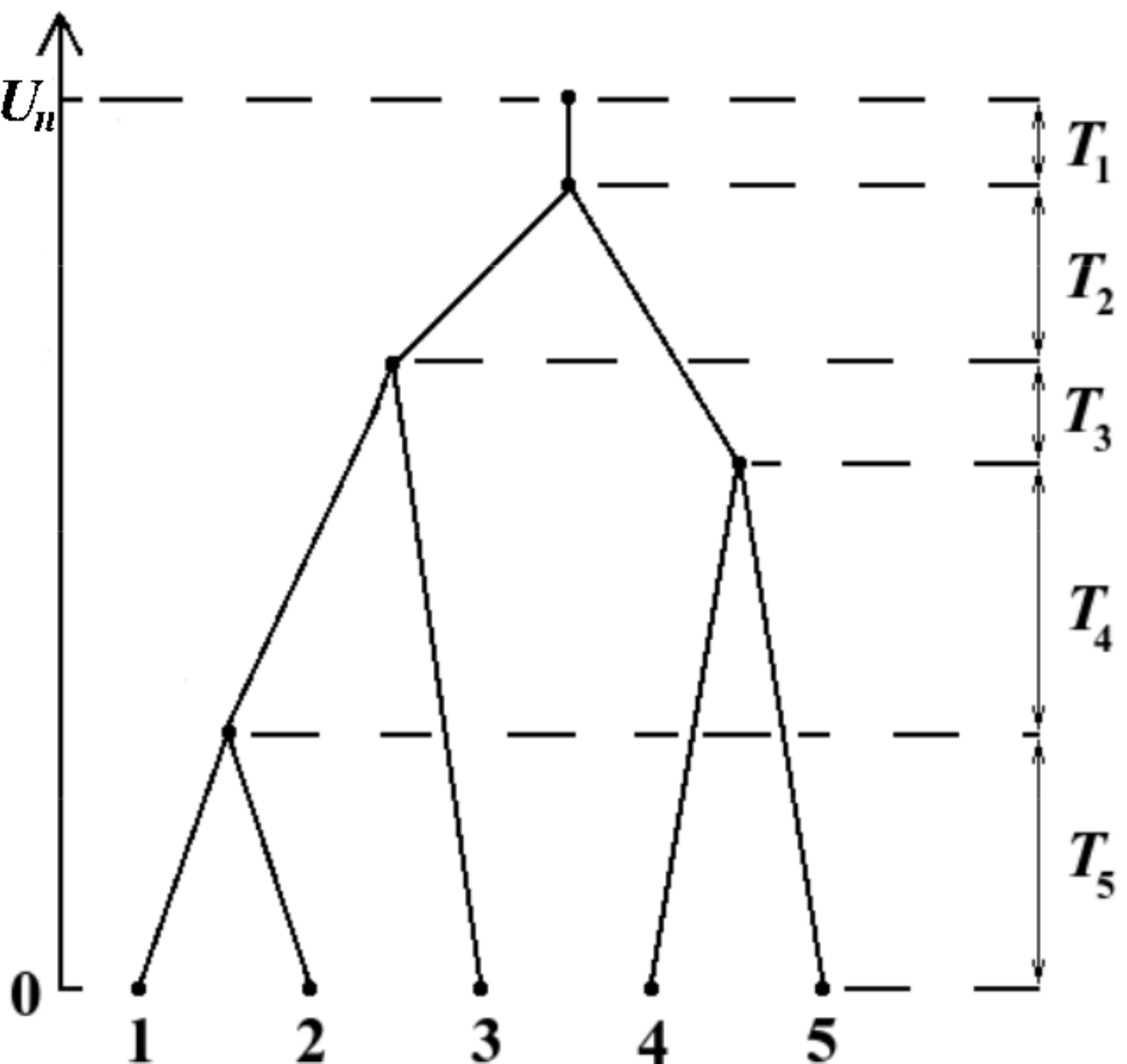}
\caption{An example of a standard Yule $n$-tree with $n=5$. 
The tree height is $U_{n}=T_{1}+\ldots+T_{n}$, where  $T_{i}$ are the times between the consecutive speciation events. The 10 pairwise 
coalescent times $\tau^{(n)}_{ij}$ for the tips of the tree are $\tau^{(5)}_{12}=T_{5}$, $\tau^{(5)}_{45}=T_{5}+T_{4}$, $\tau^{(5)}_{13}=\tau^{(5)}_{23}=T_{5}+T_{4}+T_{3}$, and $\tau^{(5)}_{14}=\tau^{(5)}_{24}=\tau^{(5)}_{34}=\tau^{(5)}_{15}=\tau^{(5)}_{25}=\tau^{(5)}_{35}=T_{5}+T_{4}+T_{3}+T_{2}$.
}\label{F3}
\end{center}
\end{figure}
Notice that the key  Proposition \ref{cypr} concerns only the first component of the evolutionary model we study -  the standard Yule $n$-tree.  For the standard Yule $n$-tree it is well known that the times between the consecutive speciation events $(T_{1},\ldots,T_n)$ are independent exponentials with parameters $(1,\ldots,n)$ respectively, see Fig. \ref{F3}. As shown in  \citet{TGer2008a}, this property corresponds to the unit rate Yule process conditioned on having  $n$  tips at the moment of observation, 
assuming that the time to the origin has the improper uniform prior
\citep[see also][]{WFel1971}.

\section{Coalescent indices of  the standard Yule $n$-tree}\label{col}

 Following the standard Yule $n$-tree from its root toward the  tips we label the consecutive splittings by indices $1,\ldots,n-1$: splitting $k$ is the vertex when $k-1$ branches turn into $k$ branches. We define three random splitting indices (as we interested in four randomly chosen tips out of $n$ available):
\begin{itemize}
\item $K_n$ is the  index of the splitting  where  two randomly chosen tips coalesce,
\item  $L_n$ be the  index of the splitting  where the first coalescent among three randomly chosen tips takes place,
\item $M_n$  be the  index of the splitting  where  the first coalescent among four randomly chosen tips takes place.
\end{itemize}
To avoid multilevel indices in the forthcoming formulae, we will often use the following notational convention 
\[KL_n:=K_{L_n},\quad LM_n:=L_{M_n},\quad KLM_n:=K_{LM_n}. \]
To illustrate these indices, turn to the Fig. \ref{F3}. If the two randomly chosen tips are $(1,2)$, then $K_n=4$. If the three randomly chosen tips are $(2,3,4)$, then $L_n=4$, $K_{L_n}=2 $. If the four randomly chosen tips are $(2,3,4,5)$, then $M_n=3$, $L_{M_n}=2$, $K_{LM_n}=1$.

The importance of these random indices comes from the following representations. Denote $U^{(n)}_k:=T_{k+1}+\ldots+T_n$ the sum of adjacent  times between splittings in the Yule tree. Clearly,
\begin{align}\label{imp}
\tau_1^{(n)}&\stackrel{d}{=}U^{(n)}_{K_n},\quad \tau_1^{(n)}\wedge\tau_2^{(n)}\stackrel{d}{=}U^{(n)}_{L_n},\quad \tau_1^{(n)}\vee\tau_2^{(n)}\stackrel{d}{=}U^{(n)}_{KL_n},\quad \tau_1^{(n)}\vee\tau_3^{(n)}\stackrel{d}{=}U^{(n)}_{KLM_n}.
\end{align}
To prove Proposition \ref{cypr} we need to know the distributions of these random splitting indices. The next two lemmata giving these distributions are proved in Section \ref{secl}.

\begin{lemma}  \label{Le1} 
 Then
\begin{align*}
P(K_{n}=k)&={n+1\over n-1}\cdot{2\over (k+1)(k+2)},\quad k=1,\ldots,n-1,\\
P(L_{n}=k)&={(n+1)(n+2)\over (n-1)(n-2)}\cdot{6(k-1)\over (k+1)(k+2)(k+3)},\quad k=2,\ldots,n-1,\\
P(M_{n}=k)&={(n+1)(n+2)(n+3)\over (n-1)(n-2)(n-3)}\cdot{12(k-1)(k-2)\over (k+1)(k+2)(k+3)(k+4)},\quad k=3,\ldots,n-1.
\end{align*}

\end{lemma}

\begin{lemma} \label{Le2}
The random numbers $K_{L_n}, L_{M_n}, K_{LM_n}$ have the following distributions
\begin{align*}
P(K_{L_n}=k)&
={n+1\over (n-1)(n-2)}\cdot{12\over (k+1)(k+2)}\Big({n+2\over k+3}-1\Big),\quad k=1,\ldots,n-2,\\
P(L_{M_n}=k)&
={(n+1)(n+2)\over (n-1)(n-2)(n-3)}\cdot{72(k-1)\over (k+1)(k+2)(k+3)}\Big({n+3\over k+4}-1\Big),\quad k=2,\ldots,n-2,
\\
P(K_{LM_n}=k)
&={n+3\over (n-1)(n-2)(n-3)}\cdot{72\over (k+1)(k+2)}\Big({(n+1)(n+2)\over (k+3)(k+4)}-1\Big)\\
&\qquad-{n+2\over (n-1)(n-2)(n-3)}\cdot{144\over (k+1)(k+2)}\Big({n+1\over k+3}-1\Big),\quad k=1,\ldots,n-3.
\end{align*}
\end{lemma}

\section{Proof of  Proposition \ref{cypr}}\label{prpr}
In view of \eqref{imp}, the harmonic numbers
\be\label{harm}
H_n=\sum_{k=1}^{n}k^{-1},\qquad \bar H_n=\sum_{k=1}^{n}k^{-2},
\ee
play an important role in our calculations as 
 \begin{align*}
\E{U^{(n)}_k}&=H_n-H_k,\\
\E{(U^{(n)}_k)^2}&=\bar H_n-\bar H_{k}+(H_n-H_{k})^2=H_n^2-2H_nH_k+\bar H_n+H_k^2-\bar H_{k}.
\end{align*}

\begin{lemma}\label{lem} 
We have
 \begin{align*}
\E{\tau_1^{(n)}}&=H_n-\E{H_{K_n}},\\
\E{(\tau_{1}^{(n)})^2}&=H_n^2-2H_n\E{H_{K_n}}+\bar H_n+\E{H_{K_n}^2-\bar H_{K_n}},\\
\E{\tau_{1}^{(n)}\tau_{2}^{(n)}}
&=H_n^2-H_n\E{{2H_{L_n}+4H_{KL_n}\over3}}+\bar H_n+\E{{2H_{L_n}H_{KL_n}+H_{KL_n}^2-2\bar H_{L_n}-\bar H_{KL_n}\over3}},\\
\E{\tau_{1}^{(n)}\tau_{3}^{(n)}}
&=H_n^2-H_n\E{{3H_{M_n}+5H_{LM_n}+10H_{KLM_n}\over9}}+\bar H_n-{1\over3}\E{\bar H_{M_n}}\\
&\qquad +{1\over9}\E{H_{M_n}H_{LM_n}+2H_{M_n}H_{KLM_n}+4H_{LM_n}H_{KLM_n}+2H_{KLM_n}^2-4\bar H_{LM_n}-2\bar H_{KLM_n}}.
\end{align*}
\end{lemma}
In view of
Lemma \ref{lem}, proven in Section \ref{secl}, the asymptotic results stated in Proposition \ref{cypr} are computed using the following  relations involving the harmonic numbers \eqref{harm}. 
\begin{lemma}\label{Lh1} 
 We have  as $n\to\infty$
 \begin{align*}
\E{H_{K_n}}&=2+O(n^{-1}\log n),\quad \E{H_{L_n}}=3+O(n^{-1}\log n),\quad \E{H_{M_n}}={11\over3}+O(n^{-1}\log n),\\
\E{H_{KL_n}}&={3\over2}+O(n^{-1}\log n),\quad \E{H_{LM_n}}={7\over3}+O(n^{-1}\log n),\quad \E{H_{KLM_n}}={4\over3}+O(n^{-1}\log n).
\end{align*}
\end{lemma}

\begin{lemma}\label{Lh2} Let $a_n\rightrightarrows a$ stand for $a_n= a+O(n^{-1}\log^2n)$ as $n\to\infty$. Then
%
\[
\begin{array}{lll}
\E{H_{K_n}^2}\rightrightarrows\frac{\pi^{2}}{3}+2,  &  \E{H_{L_n}^2}\rightrightarrows{21\over 2}, & \E{H_{M_n}^2}\rightrightarrows\frac{\pi^{2}}{3}+{211\over 18},  \\
\E{\bar H_{K_n}}\rightrightarrows\frac{\pi^{2}}{3}-2,  &  \E{\bar H_{L_n}}\rightrightarrows{3\over 2}, &   \E{\bar H_{M_n}}\rightrightarrows\frac{\pi^{2}}{3}-{31\over 18},
\end{array}
\]
and
\[
\begin{array}{lll}
 \E{H_{KL_n}^2}\rightrightarrows \frac{\pi^{2}}{2}-{9\over4}, &\E{H_{LM_n}^2}\rightrightarrows {167\over 18}-\frac{\pi^{2}}{3}, &\E{H_{KLM_n}^2}\rightrightarrows\frac{2\pi^{2}}{3}-{41\over 9},      \\
\E{\bar H_{KL_n}}\rightrightarrows\frac{\pi^{2}}{2}-{15\over4},\quad
&\E{\bar H_{LM_n}}\rightrightarrows {85\over 18}-\frac{\pi^{2}}{3},\quad
&\E{\bar H_{KLM_n}}\rightrightarrows\frac{2\pi^{2}}{3}-{49\over 9},\quad
\end{array}
\]
and 
\[
\begin{array}{ll}
\E{H_{L_n}H_{KL_n}}\rightrightarrows  {39\over4}-\frac{\pi^{2}}{2}, &\E{H_{M_n}H_{LM_n}}\rightrightarrows {221\over 18}-\frac{\pi^{2}}{3},     \\
\E{H_{M_n}H_{KLM_n}}\rightrightarrows\frac{2\pi^{2}}{3}-{14\over 9},
&\E{H_{LM_n}H_{KLM_n}}\rightrightarrows {148\over9}-{4\pi^2\over 3}.
\end{array}
\]

\end{lemma}
The proofs of the last two lemmata are given in Section \ref{h12} using the auxiliary results from Appendix \ref{App}. 

With Lemmata  \ref{lem} - \ref{Lh2} at hand, the remaining proof of  Proposition \ref{cypr} is straightforward. The first statement 
\begin{align*}
\Var{\tau_{1}^{(n)}}&=\E{(\tau_{1}^{(n)})^2}-\left(\E{\tau_1^{(n)}}\right)^2=\bar H_n+\E{H_{K_n}^2-\bar H_{K_n}} -(\E{H_{K_n}})^2\rightrightarrows{\pi^2\over6},
\end{align*}
is obtained applying  the classical relation $\bar H_n={\pi^2\over6}+O(n^{-1})$.  Further, Lemma  \ref{lem} yields
\begin{align*}
\cov{\tau_{1}^{(n)}}{\tau_{2}^{(n)}}&=\E{\tau_{1}^{(n)}\tau_{2}^{(n)}}-\left(\E{\tau_1^{(n)}}\right)^2=H_n\E{2H_{K_n}-{2H_{L_n}+4H_{KL_n}\over3}} \\
&\quad +\bar H_n+{1\over3}\E{2H_{L_n}H_{KL_n}+H_{KL_n}^2-2\bar H_{L_n}-\bar H_{KL_n}}-(\E{H_{K_n}})^2,
\end{align*}
where according to  Lemma  \ref{Lh1}
\begin{align*}
\E{2H_{K_n}-{2H_{L_n}+4H_{KL_n}\over3}} =O(n^{-1}\log n).
\end{align*}
Thus, applying Lemma \ref{Lh2} we obtain the second statement
\[\cov{\tau_{1}^{(n)}}{\tau_{2}^{(n)}}\rightrightarrows{\pi^2\over6}+{1\over3}\Big(+2\Big({39\over4}-\frac{\pi^{2}}{2}\Big) +\frac{\pi^{2}}{2}-{9\over4}-2\cdot{3\over 2} -\frac{\pi^{2}}{2}+{15\over4}\Big)-4=2-{\pi^2\over6}.\]

Finally, the third statement follows from
\begin{align*}
\cov{\tau_{1}^{(n)}}{\tau_{3}^{(n)}}&=H_n\E{2H_{K_n}-{3H_{M_n}+5H_{LM_n}+10H_{KLM_n}\over9}}+\bar H_n-(\E{H_{K_n}})^2-{1\over3}\E{\bar H_{M_n}}\\
&\quad +{1\over9}\E{H_{M_n}H_{LM_n}+2H_{M_n}H_{KLM_n}+4H_{LM_n}H_{KLM_n}+2H_{KLM_n}^2-4\bar H_{LM_n}-2\bar H_{KLM_n}}.
\end{align*}
Indeed, according to  Lemma  \ref{Lh1}
\begin{align*}
\E{2H_{K_n}-{3H_{M_n}+5H_{LM_n}+10H_{KLM_n}\over9}} =O(n^{-1}\log n).
\end{align*}
Moreover, from the following three limits
\begin{align*}
\bar H_n-(\E{H_{K_n}})^2-{1\over3}\E{\bar H_{M_n}}&\rightrightarrows
{\pi^2\over18}-{185\over54},\\
\E{H_{M_n}H_{LM_n}+2H_{M_n}H_{KLM_n}+4H_{LM_n}H_{KLM_n}}&\rightrightarrows
{1349\over18}-{13\pi^2\over 3},\\
\E{H_{KLM_n}^2-2\bar H_{LM_n}-\bar H_{KLM_n}}&\rightrightarrows 
\frac{2\pi^{2}}{3}-{77\over9},
\end{align*}
we get the stated overall limit 
$${\pi^2\over18}-{185\over54}+{1\over9}\Big({1349\over18}-{13\pi^2\over 3}\Big)+{2\over9}\Big(\frac{2\pi^{2}}{3}-{77\over9}\Big)=3-{5\pi^2\over18}.$$

\section{Proofs of Lemmata \ref{L1} - \ref{cyp}} \label{pL4}
%

\begin{proof} {\sc of Lemma \ref{L1}}.
Using the representation
\begin{align*}
D_{n}^2&=\frac{n}{n-1}\left(\frac{1}{n}\sum\limits_{i=1}^{n} Y_{i}^{2}-\overline{Y}_{n}^2\right)
=\frac{1}{n}\sum\limits_{i}Y_{i}^{2}-\frac{1}{n(n-1)}\sum_i\sum_{j\ne i}Y_{i}Y_{j}
\end{align*}
we  find that
\begin{align*}
\E{D_{n}^{4}}  &=\frac{1}{n^2} \Big(\sum\limits_{i}\E{Y_i^{4}}+\sum_i\sum_{j\ne i}\E{Y_i^{2}Y_j^2}\Big)\\
&\quad - \frac{2}{n^2(n-1)} \Big(\sum_i\sum_{j\ne i}\E{Y_{i}^{3}Y_{j}}+\sum_i\sum_{j\ne i}\E{Y_{i}Y_{j}^{3}}+\sum_i\sum_{j\ne i}\sum_{k\neq i,j}\E{Y_{i}^{2}Y_{j}Y_k}\Big)\\
&\quad + \frac{1}{n^2(n-1)^2}\Big(\sum_i\sum_{j\ne i}\E{Y_{i}^{2}Y_{j}^2}+\sum_i\sum_{j\ne i}\sum_{k\neq i,j}\E{Y_{i}^{2}Y_{j}Y_k}+\sum_i\sum_{j\ne i}\sum_{k\neq i,j}\E{Y_{i}Y_{j}^{2}Y_k}\\
&\quad\qquad\qquad\qquad\qquad\qquad\quad +\sum_i\sum_{j\ne i}\ \sum_{k\neq i,j}\ \sum_{l\neq i,j,k}\E{Y_{i}Y_{j}Y_kY_l}\Big).
\end{align*}
%
If $(W_{1},W_{2},W_{3},W_{4})$ is a random sample without replacement of four out of $n$
trait values, then
\begin{align*}
\E{W_1^{4}}  &=n^{-1} \sum\limits_{i}\E{Y_i^{4}},\\
\E{W_{1}^{3}W_{2}}&= \frac{1}{n(n-1)} \sum_{i}\sum_{j\neq i}\E{Y_{i}^{3}Y_{j}},\\
\E{W_{1}^{2}W_{2}^2}&= \frac{1}{n(n-1)}\sum_{i}\sum_{j\neq i}\E{Y_{i}^{2}Y_{j}^2},
\end{align*}
and 
\begin{align*}
\E{W_{1}^{2}W_{2}W_{3}}&= \frac{1}{n(n-1)(n-2)}\sum_{i}\sum_{j\neq i}\,\sum_{k\neq i,j}\E{Y_{i}^{2}Y_{j}Y_k},
\\
\E{W_{1}W_{2}W_{3}W_{4}}&= \frac{1}{n(n-1)(n-2)(n-3)}\sum_{i}\sum_{j\neq i}\,\sum_{k\neq i,j}\ \sum_{l\neq i,j,k}\E{Y_{i}Y_{j}Y_kY_l}.
\end{align*}
Therefore, we have
\begin{align}
\E{D_{n}^{4}}  &=n^{-1} \E{W_1^{4}}-4n^{-1}\E{W_{1}^{3}W_{2}}+\frac{n^2-2n+2}{n(n-1)}\E{W_{1}^{2}W_{2}^2}\nonumber\\
&\quad  - \frac{2(n-2)^2}{n(n-1)}\E{W_{1}^{2}W_{2}W_{3}}+ \frac{(n-2)(n-3)}{n(n-1)}\E{W_{1}W_{2}W_{3}W_{4}}.\label{stist2}
\end{align}
Since
\begin{align*}
\E{D_{n}^2}&=\E{W_1^2}-\E{W_1W_2}
,
\end{align*}
we conclude
\begin{align*}
\Var{D_{n}^{2}}  &=n^{-1} \E{W_1^{4}}-4n^{-1}\E{W_{1}^{3}W_{2}}+\cov{W_{1}^{2}}{W_{2}^2}
-\frac{n-2}{n(n-1)}\E{W_{1}^{2}W_{2}^2}\nonumber\\
&\quad  -2\cov{W_{1}^{2}}{W_{2}W_3}+ \frac{2(3n-4)}{n(n-1)}\E{W_{1}^{2}W_{2}W_{3}}\\
&\quad  +\cov{W_{1}W_2}{W_{3}W_4}- \frac{2(2n-3)}{n(n-1)}\E{W_{1}W_{2}W_{3}W_{4}}.
\end{align*}
The stated relations follow with 
\begin{align*}
B_n  &=\E{W_1^{4}}-4\E{W_{1}^{3}W_{2}}\\
&\quad -\frac{n-2}{n-1}\E{W_{1}^{2}W_{2}^2}+ \frac{2(3n-4)}{n-1}\E{W_{1}^{2}W_{2}W_{3}}- \frac{2(2n-3)}{n-1}\E{W_{1}W_{2}W_{3}W_{4}}.
\end{align*}

\end{proof}

%
%
 
 \begin{proof} {\sc of Lemma \ref{cyp}}.  Denote by $ Y_{ij}^{(n)}$ the normalized trait value of the most recent common ancestor of the tips $(i,j)$. 
Let $\mathcal Y_{ij}^{(n)}$ stand for the $\sigma$--algebra generated by the pair $(\mathcal Y_n,Y_{ij}^{(n)})$, then 
\begin{align*}
& \E{Y_i|\mathcal Y_{ij}^{(n)}}=\E{Y_j|\mathcal Y_{ij}^{(n)}}=Y_{ij}^{(n)},\\
& \var{Y_i|\mathcal Y_{ij}^{(n)}}=\var{Y_j|\mathcal Y_{ij}^{(n)}}=\tau_{ij}^{(n)},\\
 &\cov{Y_i}{Y_j|\mathcal Y_{ij}^{(n)}}=0,
\end{align*}
implying \eqref{ccov}
\begin{align*}
\covt{Y_i}{Y_j}&=\Vart{Y_{ij}^{(n)}}=U_n-\tau_{ij}^{(n)}.
\end{align*}

By Eq. (13) of \citet[][]{GBohAGol1969}, we have
\[
\cov{Z_iZ_j}{Z_kZ_l} = m_im_kc_{jl}+m_im_lc_{jk}+m_jm_kc_{il}+m_jm_lc_{ik}+c_{ik}c_{jl} + c_{il}c_{jk}
\]
for any sequence of normally distributed random values $Z_1, Z_2, \ldots$ with means $\E{Z_i}=m_{i}$ and covariances
$\cov{Z_i}{Z_j}=c_{ij}$.
In the special case with $m_i=0$ it follows
\begin{align*}
\cov{Z_iZ_j}{Z_kZ_l}& =c_{ik}c_{jl} + c_{il}c_{jk},\\
\cov{Z_i^2}{Z_jZ_k}& = 2c_{ij}c_{ik},\\
\cov{Z_i^2}{Z_j^2}& = 2c_{ij}^2.
\end{align*}
Using conditional normality of $Y_i$ and putting $c_{ij}=U_n-\tau_{ij}^{(n)}$, we derive from these relations that
 \begin{align*}
\covt{Y_i^{2}}{Y_j^{2}}&=2(U_n-\tau_{ij}^{(n)})^2,
\\
\covt{Y_i^{2}}{Y_jY_k}&=2(U_n-\tau_{ij}^{(n)})(U_n-\tau_{ik}^{(n)}),
\\
\covt{Y_iY_j}{Y_kY_l}&=(U_n-\tau_{ik}^{(n)})(U_n-\tau_{jl}^{(n)})+(U_n-\tau_{il}^{(n)})(U_n-\tau_{jk}^{(n)}).
\end{align*}
yielding in terms of \eqref{tou},
 \begin{align*}
\covt{W_1^{2}}{W_2^{2}}&=2(U_n-\tau_1^{(n)})^2,
\\
\covt{W_{1}^{2}}{W_{2}W_3}&=2(U_n-\tau_{1}^{(n)})(U_n-\tau_{2}^{(n)}),
\\
\covt{W_{1}W_{2}}{W_{3}W_{4}}&=(U_n-\tau_{2}^{(n)})(U_n-\tau_{5}^{(n)})+(U_n-\tau_{3}^{(n)})(U_n-\tau_{4}^{(n)}).
\end{align*}

 By the total covariance formula, we derive
\begin{align*}
\cov{W_{1}^{2}}{W_{2}^2}&=2\E{(U_n-\tau_1^{(n)})^2}+\var{U_n},\\
\cov{W_{1}^{2}}{W_{2}W_3}&=2\E{(U_n-\tau_{1}^{(n)})(U_n-\tau_{2}^{(n)})} +\cov{U_n}{U_n-\tau_1^{(n)}},\\
\cov{W_{1}W_{2}}{W_{3}W_{4}}&=2\E{(U_n-\tau_{1}^{(n)})(U_n-\tau_{3}^{(n)})}+\cov{U_n-\tau_{1}^{(n)}}{U_n-\tau_{3}^{(n)}}\\
&=3\E{(U_n-\tau_{1}^{(n)})(U_n-\tau_{3}^{(n)})}-\Big(\E{U_n-\tau_1^{(n)}}\Big)^2.
\end{align*}
Combining these relations we get
\begin{align*}
\cov{W_1^{2}}{W_2^{2}}&-2\cov{W_1^{2}}{W_2W_3}+\cov{W_{1}W_{2}}{W_{3}W_{4}}\\
&=2\E{(U_n-\tau_1^{(n)})^2}-4\E{(U_n-\tau_{1}^{(n)})(U_n-\tau_{2}^{(n)})}\\
&\qquad +3\E{(U_n-\tau_{1}^{(n)})(U_n-\tau_{3}^{(n)})}\\
&\qquad +\var{U_n}-2\cov{U_n}{U_n-\tau_1^{(n)}}-\Big(\E{U_n-\tau_1^{(n)}}\Big)^2.
\end{align*}
This together with
\begin{align*}
\var{U_n}-&2\cov{U_n}{U_n-\tau_1^{(n)}}-\Big(\E{U_n-\tau_1^{(n)}}\Big)^2\\
&=\E{U_n^2}-\E{U_n}^2-2\E{U_n(U_n-\tau_1^{(n)})}+2\E{U_n}\E{U_n-\tau_1^{(n)}}-\Big(\E{U_n-\tau_1^{(n)}}\Big)^2\\
&=\E{U_n^2}-2\E{U_n(U_n-\tau_1^{(n)})}-\E{\tau_1^{(n)}}^2\\
&=\E{(\tau_1^{(n)})^2}-\E{\tau_1^{(n)}}^2-\E{(U_n-\tau_1^{(n)})^2}
\end{align*}
implies the assertion of the Lemma \ref{cyp}
\begin{align*}
\cov{W_1^{2}}{W_2^{2}}&-2\cov{W_1^{2}}{W_2W_3}+\cov{W_{1}W_{2}}{W_{3}W_{4}}\\
&=\E{(U_n-\tau_1^{(n)})^2}+\E{(\tau_1^{(n)})^2}-\E{\tau_1^{(n)}}^2\\
&\quad -4\E{(U_n-\tau_{1}^{(n)})(U_n-\tau_{2}^{(n)})}  +3\E{(U_n-\tau_{1}^{(n)})(U_n-\tau_{3}^{(n)})}\\
&=\E{(U_n-\tau_1^{(n)})^2}+\var{\tau_1^{(n)}}-\E{(U_n-\tau_1^{(n)})U_n}+\E{U_n\tau_1^{(n)}}\\
&\quad -4\E{\tau_{1}^{(n)}\tau_{2}^{(n)}}  +3\E{\tau_{1}^{(n)}\tau_{3}^{(n)}}\\
&=2\var{\tau_1^{(n)}}-4\cov{\tau_{1}^{(n)}}{\tau_{2}^{(n)}}  +3\cov{\tau_{1}^{(n)}}{\tau_{3}^{(n)}}.
\end{align*}

 \end{proof}

\section{Proofs of Lemmata \ref{Le1},  \ref{Le2}, and  \ref{lem}}\label{secl}
\begin{proof} of Lemma \ref{Le1}
From the definition of $K_n$ it is easy to see that, for $k=2,\ldots,n$,
\begin{align*}
P(K_n< k-1|K_n<k)&=1-{1\over {k \choose2}}={(k+1)(k-2)\over k(k-1)}.
\end{align*}
Therefore,
\begin{align*}
P(K_n<k-1)&={(n+1)(n-2)\over n(n-1)}{n(n-3)\over (n-1)(n-2)}{(n-1)(n-4)\over (n-2)(n-3)}\cdots {(k+1)(k-2)\over k(k-1)}\\
&={(n+1)(k-2)\over (n-1)k},\\
P(K_{n}=k-1)&={(n+1)(k-1)\over (n-1)(k+1)}-{(n+1)(k-2)\over (n-1)k}={n+1\over n-1}{2\over (k+1)k}.
\end{align*}
Similarly, for $k=3,\ldots,n$,
\begin{align*}
P(L_n< k-1|L_n<k)&=1-{3\over {k \choose2}}={(k+2)(k-3)\over k(k-1)},\\
P(L_n<k-1)&={(n+2)(n-3)\over n(n-1)}{(n+1)(n-4)\over (n-1)(n-2)}{n(n-5)\over (n-2)(n-3)}\cdots {(k+2)(k-3)\over k(k-1)}\\
&={(n+2)(n+1)(k-2)(k-3)\over (n-1)(n-2)(k+1)k},\\
P(L_{n}=k-1)&={(n+2)(n+1)(k-1)(k-2)\over (n-1)(n-2)(k+2)(k+1)}-{(n+2)(n+1)(k-2)(k-3)\over (n-1)(n-2)(k+1)k}\\
&={6(n+2)(n+1)(k-2)\over (n-1)(n-2)(k+2)(k+1)k},
\end{align*}
and, for $k=4,\ldots,n$,
\begin{align*}
P(M_n< k-1|M_n<k)&=1-{6\over {k \choose2}}={(k+3)(k-4)\over k(k-1)},\\
P(M_n<k-1)&={(n+3)(n+2)(n+1)(k-2)(k-3)(k-4)\over (n-1)(n-2)(n-3)(k+2)(k+1)k},\\
P(M_{n}=k-1)
&={12(n+3)(n+2)(n+1)(k-2)(k-3)\over (n-1)(n-2)(n-3)(k+3)(k+2)(k+1)k}.
\end{align*}
\end{proof}

\begin{proof} of Lemma \ref{Le2}
Clearly,
\begin{align*}
P(K_{L_{n}}=k)&=\sum_{l=k+1}^{n-1}P(K_l=k)P(L_{n}=l)
\\
&={(n+1)(n+2)\over (n-1)(n-2)}\cdot{12\over (k+1)(k+2)}\sum_{l=k+1}^{n-1}{1\over (l+2)(l+3)}
\end{align*}
leads to the first assertion. Further,
\begin{align*}
P(L_{M_{n}}=k)&={(n+1)(n+2)(n+3)\over (n-1)(n-2)(n-3)}\\
&\times\sum_{m=k+1}^{n-1}{12(m-1)(m-2)(m+1)(m+2)\over (m+1)(m+2)(m+3)(m+4)(m-1)(m-2)}\cdot{6(k-1)\over (k+1)(k+2)(k+3)},
\end{align*}
and
\begin{align*}
P(L_{M_{n}}=k)&={(n+1)(n+2)(n+3)\over (n-1)(n-2)(n-3)}\cdot{72(k-1)\over (k+1)(k+2)(k+3)}\sum_{m=k+1}^{n-1}{1\over(m+3)(m+4)}\\
&={(n+1)(n+2)(n+3)\over (n-1)(n-2)(n-3)}\cdot{72(k-1)\over (k+1)(k+2)(k+3)}({1\over k+4}-{1\over n+3}).
\end{align*}
Finally,
\begin{align*}
P(K_{LM_{n}}=k)&=\sum_{m=k+1}^{n-2}{(n+1)(n+2)(n+3)\over (n-1)(n-2)(n-3)}\cdot{72(m-1)\over (m+1)(m+2)(m+3)(m+4)}{m+1\over m-1}\cdot{2\over (k+1)(k+2)}\\
&\qquad-\sum_{m=k+1}^{n-2}{(n+1)(n+2)\over (n-1)(n-2)(n-3)}\cdot{72(m-1)\over (m+1)(m+2)(m+3)}{m+1\over m-1}\cdot{2\over (k+1)(k+2)}\\
&={(n+1)(n+2)(n+3)\over (n-1)(n-2)(n-3)}\cdot{144\over (k+1)(k+2)}\sum_{m=k+1}^{n-2}{1\over(m+2)(m+3)(m+4)}\\
&\qquad-{(n+1)(n+2)\over (n-1)(n-2)(n-3)}\cdot{144\over (k+1)(k+2)}\sum_{m=k+1}^{n-2}{1\over(m+2)(m+3)},
\end{align*}
and therefore, it remains to use the equalities
\begin{align*}
\sum_{m=k+1}^{n-2}{1\over(m+2)(m+3)}
&={1\over k+3}-{1\over n+1},\\
\sum_{m=k+1}^{n-2}{2\over(m+2)(m+3)(m+4)}
&
={1\over (k+3)(k+4)}-{1\over (n+1)(n+2)}.
\end{align*}
\end{proof}

 \begin{proof} {\sc of Lemma \ref{lem}}.  
We have 
\begin{align*}
\E{\tau_{1}^{(n)}}&=\E{U^{(n)}_{K_n}}=H_n-\E{H_{K_n}},
\end{align*}
and
\begin{align*}
\E{(\tau_{1}^{(n)})^2}&=\E{(U^{(n)}_{K_n})^2}=\bar H_n+H_n^2-\E{\bar H_{K_n}+2H_{K_n}H_n-H_{K_n}^2}.
\end{align*}
Further, using \eqref{imp} and $U^{(n)}_{KL_n}=U^{(n)}_{L_n}+U^{(L_n)}_{KL_n}$, we get
\begin{align*}
\E{\tau_{1}^{(n)}\tau_{2}^{(n)}}&={1\over3}\E{(U^{(n)}_{KL_n})^2}+{2\over3}\E{U^{(n)}_{KL_n}U^{(n)}_{L_n}}\\
&=\E{(U^{(n)}_{L_n})^2}+{4\over3}\E{U^{(n)}_{L_n}U^{(L_n)}_{KL_n}}+{1\over3}\E{(U^{(L_n)}_{KL_n})^2}.
\end{align*}
Thus
\begin{align*}
\E{\tau_{1}^{(n)}\tau_{2}^{(n)}}&=H_n^2-2H_n\E{H_{L_n}}+\bar H_n+\E{H_{L_n}^2-\bar H_{L_n}}+{4\over3}\E{(H_n-H_{L_n})(H_{L_n}-H_{KL_n})}\\
&\quad +{1\over3}\E{H_{L_n}^2-2H_{L_n}H_{KL_n}+\bar H_{L_n}+H_{KL_n}^2-\bar H_{KL_n}}
\\
&=H_n^2-H_n\E{{2H_{L_n}+4H_{KL_n}\over3}}+\bar H_n+\E{{2H_{L_n}H_{KL_n}+H_{KL_n}^2-2\bar H_{L_n}-\bar H_{KL_n}\over3}
}.
\end{align*}
Finally, for two pairs of sampled tips, we have three coalescent events to consider:
going from four to three selected nodes, $4\to3$, 
going from three to two selected nodes, $3\to2$, 
and going from two to one selected nodes, $2\to1$. 
The 
coalescent $4\to3$ holds across the two pairs with probability ${4\over{4\choose2}}={2\over3}$ and within a pair with probability ${1\over3}$. Given the former outcome, the coalescent $3\to2$ holds again across the pairs with probability ${1\over3}$ and within a pair with probability ${2\over3}$. Otherwise, the coalescent $3\to2$ holds across the pairs with probability ${2\over3}$ and within the second pair with probability ${1\over3}$. The four possibilities (${2\over3}\times{1\over3}$, ${2\over3}\times{2\over3}$, ${1\over3}\times{2\over3}$, ${1\over3}\times{1\over3}$) produce the following four terms in 
\begin{align*}
\E{\tau_{1}^{(n)}\tau_{3}^{(n)}}={2\over9}\E{(U^{(n)}_{KLM_n})^2}
+{4\over9}\E{U^{(n)}_{LM_n}U^{(n)}_{KLM_n}}+{2\over9}\E{U^{(n)}_{M_n}U^{(n)}_{KLM_n}}+{1\over9}\E{U^{(n)}_{M_n}U^{(n)}_{LM_n}}.
\end{align*}
It follows,
\begin{align*}
\E{\tau_{1}^{(n)}\tau_{3}^{(n)}}&=\E{(U^{(n)}_{M_n})^2+{2\over9}(U^{(M_n)}_{KLM_n})^2
+{10\over9}U^{(n)}_{M_n}U^{(M_n)}_{KLM_n}+{4\over9}U^{(M_n)}_{LM_n}U^{(M_n)}_{KLM_n}+{5\over9}U^{(n)}_{M_n}U^{(M_n)}_{LM_n}}.
\end{align*}
Using the representation for $\E{\tau_{1}^{(M_n)}\tau_{2}^{(M_n)}}$,
\begin{align*}
&{1\over3}\E{(U^{(M_n)}_{KLM_n})^2}+{2\over3}\E{U^{(M_n)}_{KLM_n}U^{(M_n)}_{LM_n}}\\
&\quad=\E{H_{M_n}^2-H_{M_n}{2H_{LM_n}+4H_{KLM_n}\over3}+\bar H_{M_n}+{2H_{LM_n}H_{KLM_n}+H_{KLM_n}^2-2\bar H_{LM_n}-\bar H_{KLM_n}\over3}},
\end{align*}
we can write
\begin{align*}
\E{\tau_{1}^{(n)}\tau_{3}^{(n)}}&=H_n^2-\E{2H_nH_{M_n}+\bar H_{M_n}-H_{M_n}^2}
+\bar H_n+{2\over3}\E{H_{M_n}^2-H_{M_n}{2H_{LM_n}+4H_{KLM_n}\over3}+\bar H_{M_n}}\\
&\quad+{2\over3}\E{{2H_{LM_n}H_{KLM_n}+H_{KLM_n}^2-2\bar H_{LM_n}-\bar H_{KLM_n}\over3}}\\
&\quad+{5\over9}\E{(H_n-H_{M_n})(H_{M_n}-H_{LM_n})}+{10\over9}\E{(H_n-H_{M_n})(H_{M_n}-H_{KLM_n})},
\end{align*}
which after a rearrangement gives the last statement.
 \end{proof}

\section{Proof of Lemmata \ref{Lh1} - \ref{Lh2}}\label{h12}
In this section we will often use the elementary relations of the  following type
\begin{align}
{6\over (k+1)(k+2)(k+3)(k+4)}
&={1\over (k+1)(k+2)}-{2\over (k+2)(k+3)}+{1\over (k+3) (k+4)},\label{64}
\\{(k-1)(k-2)\over (k+1)(k+2)(k+3)(k+4)}
&={1\over (k+1)(k+2)}-{5\over (k+2)(k+3)}+{5\over (k+3) (k+4)},\label{61}\\
{6k\over (k+1)(k+2)(k+3)(k+4)}&=-{1\over (k+1)(k+2)}+{5\over (k+2)(k+3)}-{4\over (k+3)(k+4)},\label{62}\\
{k(k-5)\over (k+1)(k+2)(k+3)(k+4)}&={1\over (k+1)(k+2)}+{6\over (k+3)(k+4)}-{6\over(k+2)(k+3)},\label{65}
\end{align}
valid for all $k\ge1$.

\begin{proof} of Lemma \ref{Lh1}. The first three stated relations  are obtained using Lemmata \ref{Le1} and \ref{L4}. Equalities
 \begin{align*}
\E{H_{K_n}}&={n+1\over n-1}\sum_{k=1}^{n-1}{2H_{k}\over  (k+1)(k+2)}=\frac{2(n-H_{n})}{n-1},\\
\E{H_{L_n}}&={(n+1)(n+2)\over (n-1)(n-2)}\sum_{k=2}^{n-1}{6(k-1)H_k\over (k+1)(k+2)(k+3)}\\
&={6(n+1)(n+2)\over (n-1)(n-2)}\sum_{k=1}^{n-1}
\Big({2H_k\over (k+2)(k+3)}-{H_k\over (k+1)(k+2)}\Big)
={3n(n+1)-6nH_n\over (n-1)(n-2)},
\end{align*}
give the first and the second stated relations, and the third one follows from 
\begin{align*}
\E{H_{M_n}}&={(n+1)(n+2)(n+3)\over (n-1)(n-2)(n-3)}\sum_{k=3}^{n-1}{12(k-1)(k-2)H_k\over (k+1)(k+2)(k+3)(k+4)}\\
&\stackrel{\eqref{61}}{=}{12(n+1)(n+2)(n+3)\over (n-1)(n-2)(n-3)}\sum_{k=1}^{n-1}
\Big({H_k\over (k+1)(k+2)}-{5H_k\over (k+2)(k+3)}+{5H_k\over (k+3)(k+4)}\Big)
\\&={{11\over3}n^3+30n^2+{391\over3}n-20-12(n^2+1)H_n\over (n-1)(n-2)(n-3)}.
\end{align*}
The second three stated relations are obtained similarly  using Lemmata \ref{Le2} and \ref{L4}. Indeed, 
\begin{align*}
\E{H_{KL_n}}&={(n+1)(n+2)\over (n-1)(n-2)}\sum_{k=1}^{n-1}{12H_k\over (k+1)(k+2)(k+3)}-{n+1\over (n-1)(n-2)}\sum_{k=1}^{n-1}{12H_k\over (k+1)(k+2)}\\
&={6(n+1)(n+2)\over (n-1)(n-2)}\Big(\sum_{k=1}^{n-1}{H_k\over (k+1)(k+2)}-\sum_{k=1}^{n-1}{H_k\over (k+2)(k+3)}\Big)-{12(n-H_n)\over (n-1)(n-2)}\\
&={6n(n-H_{n})\over (n-1)(n-2)}-{6(3n^2+5n-4(n+1)H_n)\over 4(n-1)(n-2)}
={6H_{n}\over (n-1)(n-2)}+\frac{3n(n-5)}{2(n-1)(n-2)},
\end{align*}
implying $\E{H_{KL_n}}={3\over2}+O(n^{-1}\log n)$.
Furthermore, 
\begin{align*}
\E{H_{LM_n}}&=\sum_{k=2}^{\infty}{72(k-1)H_k\over (k+1)(k+2)(k+3)(k+4)}+O(n^{-1}\log n)\stackrel{\eqref{64}\eqref{62}}{=}{7\over 3}+O(n^{-1}\log n),\\
\E{H_{KLM_n}}&=\sum_{k=1}^{\infty}{72H_k\over (k+1)(k+2)(k+3)(k+4)}+O(n^{-1}\log n)\stackrel{\eqref{64}}{=}{4\over 3}+O(n^{-1}\log n).
\end{align*}

\end{proof}

\begin{proof} of Lemma \ref{Lh2}. The stated relations are obtained using Lemmas \ref{Le1}, \ref{Le2}, \ref{L4}, \ref{L2}. Firstly,
\begin{align*}
\E{H_{K_n}^2}&={n+1\over n-1}\sum_{k=1}^{n-1}{2H_k^2\over (k+1)(k+2)}
={2(n+1)\over n-1}\Big(\bar H_{n}+\frac{n-H_{n}^2-2H_n}{n+1}
\Big)\\
&={2(\bar H_n(n+1)+n-H_{n}^2-2H_n)\over n-1}\rightrightarrows\frac{\pi^{2}}{3}+2.
\end{align*}
Similarly, we have
\begin{align*}
\E{\bar H_{K_n}}&={n+1\over n-1}\sum_{k=1}^{n-1}{2\bar H_{k}\over (k+1)(k+2)}
={2(n+1)\over n-1}\cdot\frac{n\bar H_{n}-n}{n+1} 
={2n\bar H_n-2n\over n-1}\rightrightarrows\frac{\pi^{2}}{3}-2.
\end{align*}
Observe that the limit is $\sum_{k=1}^{\infty}{2\bar H_{k}\over (k+1)(k+2)}$. In the same manner we obtain
\begin{align*}
\E{H_{KL_n}^2}&\rightrightarrows \sum_{k=1}^{\infty}{12H_{k}^2\over (k+1)(k+2)(k+3)}=\sum_{k=1}^{\infty}{6H_{k}^2\over (k+1)(k+2)}-\sum_{k=1}^{\infty}{6H_{k}^2\over (k+2)(k+3)}=\frac{\pi^{2}}{2}-{9\over4},\\
\E{\bar H_{KL_n}}&\rightrightarrows \sum_{k=1}^{\infty}{12\bar H_{k}\over (k+1)(k+2)(k+3)}=\sum_{k=1}^{\infty}{6\bar H_{k}\over (k+1)(k+2)}-\sum_{k=1}^{\infty}{6\bar H_{k}\over (k+2)(k+3)}=\frac{\pi^{2}}{2}-{15\over4}.
\end{align*}
Using the decomposition \eqref{61}
we find
\begin{align*}
\E{H_{M_n}^2}&\rightrightarrows \sum_{k=3}^{\infty}{12(k-1)(k-2)H_k^2\over (k+1)(k+2)(k+3)(k+4)}
=\frac{\pi^{2}}{3}+{211\over 18},\\
\E{\bar H_{M_n}}&\rightrightarrows \sum_{k=3}^{\infty}{12(k-1)(k-2)\bar H_{k}\over (k+1)(k+2)(k+3)(k+4)}
=\frac{\pi^{2}}{3}-{31\over 18}.
\end{align*}
Using the difference between \eqref{62} and \eqref{64}
we find
\begin{align*}
\E{H_{LM_n}^2}&\rightrightarrows \sum_{k=1}^{\infty}{72(k-1)H_k^2\over (k+1)(k+2)(k+3)(k+4)}
={167\over 18}-\frac{\pi^{2}}{3},\\
\E{\bar H_{LM_n}}&\rightrightarrows \sum_{k=1}^{\infty}{72(k-1)\bar H_{k}\over (k+1)(k+2)(k+3)(k+4)}
={85\over 18}-\frac{\pi^{2}}{3}.
\end{align*}
Using \eqref{64}
we find
\begin{align*}
\E{H_{KLM_n}^2}&\rightrightarrows \sum_{k=1}^{\infty}{72H_k^2\over (k+1)(k+2)(k+3)(k+4)}
=\frac{2\pi^{2}}{3}-{41\over 9},\\
\E{\bar H_{KLM_n}}&\rightrightarrows \sum_{k=1}^{\infty}{72\bar H_{k}\over (k+1)(k+2)(k+3)(k+4)}
=\frac{2\pi^{2}}{3}-{49\over 9}.
\end{align*}

Since
$
\E{H_{K_n}}=\frac{2(n-H_{n})}{n-1}$, we have
\begin{align*}
\E{H_{L_n}H_{KL_n}}
&\rightrightarrows \sum_{m=2}^{\infty}{6(m-1)H_m\over (m+1)(m+2)(m+3)}{2(m-H_m)\over m-1}\\
&=\sum_{m=1}^{\infty}{12mH_m\over (m+1)(m+2)(m+3)}-\sum_{m=1}^{\infty}{12H_m^2\over (m+1)(m+2)(m+3)}\\
&={15\over2}-(\frac{\pi^{2}}{2}-{9\over4})={39\over4}-\frac{\pi^{2}}{2},
\end{align*}
where we use the following corollary of Lemma \ref{L4}
\begin{align*}
\sum_{m=1}^{\infty}&{2mH_m\over (m+1)(m+2)(m+3)}=\sum_{m=1}^{\infty}{3H_m\over (m+2)(m+3)}-\sum_{m=1}^{\infty}{H_m\over (m+1)(m+2)}
={5\over4}.
\end{align*}
Similarly,
\begin{align*}
\E{H_{LM_n}H_{KLM_n}}
&\rightrightarrows \sum_{m=2}^{\infty}{72(m-1)H_m\over (m+1)(m+2)(m+3)(m+4)}{2(m-H_m)\over m-1}\\
&=\sum_{m=1}^{\infty}{144mH_m\over (m+1)(m+2)(m+3)(m+4)}-\sum_{m=1}^{\infty}{144H_m^2\over (m+1)(m+2)(m+3)(m+4)},
\end{align*}
where
\begin{align*}
\sum_{k=1}^{\infty}{144kH_k\over (k+1)(k+2)(k+3)(k+4)}&\stackrel{\eqref{62}}{=}24(-1+15/4-22/9)=
{22\over3},\\
\sum_{k=1}^{\infty}{144H_k^2\over (k+1)(k+2)(k+3)(k+4)}&\stackrel{\eqref{64}}{=}{4\pi^2\over3}-{82\over9},
\end{align*}
so that
$\E{H_{LM_n}H_{KLM_n}}
\rightrightarrows {148\over9}-{4\pi^2\over 3}$.

Further, in view of
$
\E{H_{L_n}}={3n(n+1)-6nH_n\over (n-1)(n-2)}$,
the limit for $\E{H_{M_n}H_{LM_n}}$ can be computed as
\begin{align*}
\E{H_{M_n}H_{LM_n}}
&\rightrightarrows \sum_{m=1}^{\infty}{12(m-1)(m-2)H_m\over (m+1)(m+2)(m+3)(m+4)}{3m(m+1)-6mH_m\over (m-1)(m-2)}\\
&=\sum_{m=1}^{\infty}{36mH_m\over (m+2)(m+3)(m+4)}-\sum_{m=1}^{\infty}{72mH_m^2\over (m+1)(m+2)(m+3)(m+4)},
\end{align*}
where
\begin{align*}
\sum_{m=1}^{\infty}{6mH_m^2\over (m+1)(m+2)(m+3)(m+4)}&\stackrel{\eqref{62}}{=}{\pi^2\over 36}+{85\over216},\\
\sum_{m=1}^{\infty}{mH_m\over (m+2)(m+3)(m+4)}&=\sum_{m=1}^{\infty}{2H_m\over (m+3)(m+4)}-\sum_{m=1}^{\infty}{H_m\over (m+2)(m+3)}
={17\over36},
\end{align*}
yielding
$\E{H_{M_n}H_{LM_n}}
\rightrightarrows {221\over18}-{\pi^2\over 3}$.
Finally, from
\begin{align*}
\E{H_{KL_m}}&={6H_{m}\over (m-1)(m-2)}+\frac{3m(m-5)}{2(m-1)(m-2)}
\end{align*}
we get
\begin{align*}
\E{H_{M_n}H_{KLM_n}}
&\rightrightarrows \sum_{m=1}^{\infty}{12(m-1)(m-2)H_m\over (m+1)(m+2)(m+3)(m+4)}\Big({6H_{m}\over (m-1)(m-2)}+\frac{3m(m-5)}{2(m-1)(m-2)}\Big)\\
&=\sum_{m=1}^{\infty}{72H_m^2\over (m+1)(m+2)(m+3)(m+4)}+\sum_{m=1}^{\infty}{18m(m-5)H_m\over (m+1)(m+2)(m+3)(m+4)},
\end{align*}
where
\begin{align*}
\sum_{m=1}^{\infty}{72H_m^2\over (m+1)(m+2)(m+3)(m+4)}
&\stackrel{\eqref{64}}{=}{2\pi^2\over 3}-{41\over9},\\
\sum_{m=1}^{\infty}{m(m-5)H_m\over (m+1)(m+2)(m+3)(m+4)}&\stackrel{\eqref{65}}{=}{1\over6},
\end{align*}
so that
$
 \E{H_{M_n}H_{KLM_n}}\rightrightarrows {2\pi^2\over3}-{14\over9}$.
\end{proof}

\section*{Acknowledgements}
The research of Serik Sagitov was supported by the Swedish Research Council grant 621-2010-5623. 
Krzysztof Bartoszek was supported by the Centre for Theoretical Biology at the University of Gothenburg, 
Svenska Institutets \"Ostersj\"osamarbete scholarship nr. 11142/2013,
Stiftelsen f\"or Vetenskaplig Forskning och Utbildning i Matematik
(Foundation for Scientific Research and Education in Mathematics), 
Knut and Alice Wallenbergs travel fund, Paul and Marie Berghaus fund, the Royal Swedish Academy of Sciences,
and Wilhelm and Martina Lundgrens research fund.

\bibliographystyle{plainnat}
\bibliography{SagitovBartoszek}

\appendix
\section{Auxiliary results involving harmonic numbers}\label{App}
Some of the following results can be found in \citet{VAda1997} and \citet{ASofo2011,ASofo2012,ASofo2013}.

\begin{lemma}\label{L4}
We have 
$$\sum\limits_{k=1}^{n-1}\frac{H_{k}}{k(k+1)}=\bar H_n-\frac{H_n}{n},\quad \sum\limits_{k=1}^\infty\frac{H_{k}}{k(k+1)}=\frac{\pi^2}{6},$$
and for $m\ge1$,
\begin{align*}
\sum\limits_{k=1}^{n-1}\frac{H_{k}}{(k+m)(k+m+1)}&
=\frac{H_{m}}{m}-\frac{H_{n+m}-H_{n}}{m}-\frac{H_{n}}{n+m},\\
\sum\limits_{k=1}^{n-1}\frac{\bar H_{k}}{(k+m)(k+m+1)}&=\frac{n\bar H_{n}}{(n+m)m} - \frac{H_{m}}{m^{2}}+\frac{H_{n+m}-H_n}{m^{2}},
\end{align*}
so that
\begin{align*}
\sum\limits_{k=1}^{\infty}\frac{H_{k}}{(k+m)(k+m+1)}&={H_m\over m},\qquad \sum\limits_{k=1}^{\infty}\frac{\bar H_{k}}{(k+m)(k+m+1)} =\frac{\pi^{2}}{6m} - \frac{H_{m}}{m^{2}}.
\end{align*}
In particular,
\begin{align*}
&\sum\limits_{k=1}^{n-1}\frac{H_{k}}{(k+1)(k+2)}= \frac{n-H_{n}}{n+1},\quad\sum\limits_{k=1}^{n-1}\frac{H_{k}}{(k+2)(k+3)}
=\frac{3n^2+5n-4(n+1)H_{n}}{4(n+1)(n+2)},\\
&\sum\limits_{k=1}^{n-1}\frac{\bar H_{k}}{(k+1)(k+2)}=\frac{n\bar H_{n}-n}{n+1},\quad 
\sum\limits_{k=1}^{n-1}\frac{\bar H_{k}}{(k+2)(k+3)}
={n\bar H_{n}\over2(n+2)}-{3n^2+5n\over8(n+1)(n+2)},
\end{align*}
and
\begin{align*}
&\sum\limits_{k=1}^\infty\frac{H_{k}}{(k+1)(k+2)}= 1,\quad
\sum\limits_{k=1}^\infty\frac{H_{k}}{(k+2)(k+3)}=\frac{3}{4},\quad
\sum\limits_{k=1}^\infty\frac{H_{k}}{(k+3)(k+4)}=\frac{11}{18},\\
&\sum\limits_{k=1}^\infty\frac{\bar H_{k}}{(k+1)(k+2)}=\frac{\pi^2}{6}-1,\quad 
\sum\limits_{k=1}^\infty\frac{\bar H_{k}}{(k+2)(k+3)}=\frac{\pi^2}{12}-\frac{3}{8},\quad 
\sum\limits_{k=1}^\infty\frac{\bar H_{k}}{(k+3)(k+4)}=\frac{\pi^2}{18}-\frac{11}{54}.
\end{align*}
\end{lemma}

\begin{proof}
Clearly, 
\begin{align*}
\sum\limits_{k=1}^{n-1}\frac{H_{k}}{k(k+1)}&=\sum_{k=1}^{n-1}\frac{1}{k(k+1)}\sum_{i=1}^{k}{1\over i}
=\sum_{i=1}^{n-1}{1\over i}\Big(\frac{1}{i}-\frac{1}{n}\Big)=\bar H_{n}-
\frac{H_n}{n}.
\end{align*}
Similarly for $m\ge1$, we have
\begin{align*}
\sum\limits_{k=1}^{n-1}\frac{H_{k}}{(k+m)(k+m+1)}
&=\sum_{i=1}^{n-1}{1\over i}\Big(\frac{1}{i+m}-\frac{1}{n+m}\Big)=\sum_{i=1}^{n}{1\over i}\Big(\frac{1}{i+m}-\frac{1}{n+m}\Big)\\
&= \frac{1}{m}\left(\sum_{i=1}^{n}{1\over i}-\sum_{i=1}^{n}\frac{1}{i+m}\right)-\frac{H_n}{n+m}= \frac{1}{m}\left(H_{n}-H_{n+m}+H_{m}\right)-\frac{H_n}{n+m},
\end{align*}
and
\begin{align*}
\sum\limits_{k=1}^{n-1}&\frac{\bar H_{k}}{(k+m)(k+m+1)}=\sum_{i=1}^{n}{1\over i^2}\Big(\frac{1}{i+m}-\frac{1}{n+m}\Big)= \frac{1}{m}\left(\sum_{i=1}^{n}{1\over i^2}
-\sum_{i=1}^{n}\frac{1}{i(i+m)}\right)-\frac{\bar H_n}{n+m}\\
&= \frac{1}{m}\left(\bar H_n- \frac{1}{m}\left(H_n+H_{m}-H_{n+m}\right)\right)-\frac{\bar H_n}{n+m} 
= \frac{n\bar H_{n}}{(n+m)m} - \frac{H_{m}}{m^{2}}+
\frac{H_{n+m}-H_n}{m^{2}}.
\end{align*}
 
\end{proof}

\begin{lemma}\label{L2}
We have
\begin{align*}
\sum\limits_{k=1}^{n-1}&\frac{H_{k}^2}{(k+1)(k+2)}
=\bar H_{n}+\frac{n-H_{n}^2-2H_n}{n+1},\\
\sum\limits_{k=1}^{n-1}&\frac{H_{k}^2}{(k+2)(k+3)}
=\frac{\bar H_{n}}{2}+{11n^2+21n\over8(n+1)(n+2)}-{H_n(2n+3)\over (n+1)(n+2)}-\frac{H_{n}^2}{n+2},
\end{align*}
and generally for $m\ge1$, 
\begin{align*}
\sum\limits_{k=1}^{\infty}\frac{H_{k}^2}{(k+m)(k+m+1)}&=\frac{1}{m}\Big({\pi^2\over 6}+H_m^2+\bar H_m-\frac{H_{m}}{m}\Big).
\end{align*}
In particular,
\begin{align*}
\sum\limits_{k=1}^{\infty}\frac{H_{k}^2}{(k+1)(k+2)}&={\pi^2\over 6}+1,\quad
\sum\limits_{k=1}^{\infty}\frac{H_{k}^2}{(k+2)(k+3)}={\pi^2\over 12}+{11\over8},\quad
\sum\limits_{k=1}^{\infty}\frac{H_{k}^2}{(k+3)(k+4)}={\pi^2\over 18}+{37\over27}.
\end{align*}
\end{lemma}

\begin{proof} 
For $m\ge1$,
%
\begin{align*}
\sum\limits_{k=1}^{n-1}\frac{H_{k}^2}{(k+m)(k+m+1)}&=\sum_{k=1}^{n-1}\frac{H_k}{(k+m)(k+m+1)}\sum_{i=1}^{k}{1\over i}=\sum_{i=1}^{n-1}{1\over i}\sum_{k=i}^{n-1}\frac{H_k}{(k+m)(k+m+1)}
\\
&=\sum_{i=1}^{n-1}{1\over i}\left(\frac{H_{m}}{m}-\frac{H_{n+m}-H_{n}}{m}-\frac{H_{n}}{n+m}-\frac{H_{m}}{m}+\frac{H_{i+m}-H_{i}}{m}+\frac{H_{i}}{i+m}\right)\\
&=\sum_{i=1}^{n-1}{1\over i}\left(\frac{H_{i+m}-H_{i}}{m}+\frac{H_{i}}{i+m}\right)-\frac{H_{n-1}(H_{n+m}-H_{n})}{m}-\frac{H_{n-1}H_{n}}{n+m}\\
&=\frac{1}{m}\sum_{i=1}^{n-1}\left({H_{i+m}\over i}-\frac{H_{i}}{i+m}\right)-\frac{H_{n-1}(H_{n+m}-H_{n})}{m}-\frac{H_{n-1}H_{n}}{n+m}.
\end{align*}
Observe that
\begin{align*}
\sum_{i=1}^{n-1}\left({H_{i+1}\over i}-\frac{H_{i}}{i+1}\right)&=\sum_{i=1}^{n-1}\left({H_{i+1}\over i}-{H_i\over i}\right)+\bar H_n-{H_n\over n}=\bar H_{n}+1-{1\over n}-{H_{n}\over n},
\end{align*}
and for $k\ge2$,
\begin{align*}
\sum_{i=1}^{n-1}\left({H_{i+k}\over i}-\frac{H_{i}}{i+k}\right)&-\sum_{i=1}^{n-1}\left({H_{i+k-1}\over i}-\frac{H_{i}}{i+k-1}\right)
= \sum_{i=1}^{n-1}{1\over i(i+k)}+\sum_{i=1}^{n-1}\frac{H_{i}}{(i+k)(i+k-1)}
\\
&=\frac{1}{k}\left(H_{n-1}+H_{k}-H_{n+k-1}\right)+\frac{H_{k-1}}{k-1}-\frac{H_{n+k-1}-H_{n}}{k-1}-\frac{H_{n}}{n+k-1}
\\
&=\frac{H_{k}}{k}+\frac{H_{k-1}}{k-1}-\frac{H_{n+k-1}-H_{n-1}}{k}-\frac{H_{n+k-1}-H_{n}}{k-1}-\frac{H_{n}}{n+k-1}.
\end{align*}
It follows
\begin{align*}
\sum_{i=1}^{n-1}\left({H_{i+m}\over i}-\frac{H_{i}}{i+m}\right)&=\bar H_{n}+1-{1\over n}-{H_{n}\over n}
\\
&+\sum_{k=2}^{m}\left(\frac{H_{k}}{k}+\frac{H_{k-1}}{k-1}-\frac{H_{n+k-1}-H_{n-1}}{k}-\frac{H_{n+k-1}-H_{n}}{k-1}-\frac{H_{n}}{n+k-1}\right)\\
&=\bar H_{n}-{1\over n}-{H_{n}\over n}+2\sum_{k=1}^{m}\frac{H_k}{k}-{H_m\over m}
\\
&-\sum_{k=2}^{m}\left(\frac{H_{n+k-1}-H_{n-1}}{k}+\frac{H_{n+k-1}-H_{n}}{k-1}\right)-H_{n}(H_{n+m-1}-H_n).
\end{align*}
Using the classical relation $2\sum_{k=1}^{m}\frac{H_k}{k} =H_m^{2}+\bar H_{m}$ which follows from 
\[ \sum_{k=1}^{m}\frac{H_k}{k} =\sum_{k=1}^{m}\frac{1}{k} \sum_{i=1}^{k}\frac{1}{i}
=\sum_{i=1}^{m}\frac{1}{i}\sum_{k=i}^{m}\frac{1}{k} =\sum_{i=1}^{m}\frac{H_m-H_{i-1}}{i}
= H_m^{2}+\bar H_{m}-\sum_{k=1}^{m}\frac{H_k}{k},\]
we get
\begin{align*}
&\sum_{i=1}^{n-1}\left({H_{i+m}\over i}-\frac{H_{i}}{i+m}\right)=\bar H_{n}-{1\over n}+H_m^{2}+\bar H_{m}-{H_m\over m}
\\
&\qquad \qquad \qquad \qquad -\sum_{k=2}^{m}\left(\frac{H_{n+k-1}-H_{n-1}}{k}+\frac{H_{n+k-1}-H_{n}}{k-1}\right)-H_{n}(H_{n+m-1}-H_{n-1})\\
&\quad =\bar H_{n}+H_m^{2}+\bar H_{m}-{H_m\over m}-\sum_{k=1}^{m}\frac{H_{n+k-1}-H_{n-1}}{k}-\sum_{k=1}^{m-1}\frac{H_{n+k}-H_{n}}{k}-H_{n}(H_{n+m-1}-H_{n-1}),
\end{align*}
Thus
\begin{align*}
\sum\limits_{k=1}^{n-1}&\frac{H_{k}^2}{(k+m)(k+m+1)}
=\frac{1}{m}\left(\bar H_{n}+H_m^{2}+\bar H_{m}-{H_m\over m}\right)-\frac{H_{n-1}(H_{n+m}-H_{n})}{m}-\frac{H_{n-1}H_{n}}{n+m}\\
&-\frac{1}{m}\sum_{k=1}^{m}\frac{H_{n+k-1}-H_{n-1}}{k}-\frac{1}{m}\sum_{k=1}^{m-1}\frac{H_{n+k}-H_{n}}{k}-{H_n(H_{n+m-1}-H_{n-1})\over m}.\end{align*}
To finish the proof it remains to observe that $\frac{H_{n+k}-H_{n}}{k}\to0$ as $n\to\infty$ for any fixed $k$.

\end{proof}


\end{document}